\documentclass[12pt,a4paper]{article}
\usepackage[dvipsnames]{xcolor}
\usepackage{jheppub}
\usepackage{amssymb}
\usepackage{amsmath}
\usepackage{amsthm}
\usepackage{mathtools}
\usepackage{amsfonts}
\usepackage{dsfont}
\usepackage{mathrsfs}
\usepackage{tikz}
\usepackage{orcidlink}

\hypersetup{
    pdfencoding=unicode,
    colorlinks=true,
    urlcolor=Maroon,
    linkcolor=RoyalBlue,
    citecolor=Maroon,
    pdftitle={The Super Virasoro Minimal String},
    pdfauthor={Lorenz Eberhardt},
    pdfdisplaydoctitle=true,
    pdfstartview=FitH,
    linktocpage=true
}

\newtheorem{lemma}{Lemma}[section]
\newtheorem{proposition}{Proposition}[section]

\allowdisplaybreaks[1]

\DeclareMathOperator\rk{rk}
\DeclareMathOperator\Aut{Aut}
\DeclareMathOperator\td{td}
\DeclareMathOperator\ch{ch}
\DeclareMathOperator\std{std}

\newcommand\Res{\mathop{\mathrm{Res}}}

\newcommand{\be}{\begin{equation}}
\newcommand{\ee}{\end{equation}}
\newcommand{\ba}{\begin{align}}
\newcommand{\ea}{\end{align}}
\newcommand{\bd}{\begin{tikzpicture}}
\newcommand{\ed}{\end{tikzpicture}}

\renewcommand\d{\mathrm{d}}

\newcommand\CC{\mathbb{C}}
\newcommand\ZZ{\mathbb{Z}}
\newcommand\RR{\mathbb{R}}

\newcommand\CP{\mathbb{CP}}

\newcommand\PSL{\mathrm{PSL}}

\newcommand\Sp{\mathrm{Sp}}

\newcommand\SL{\mathrm{SL}}

\newcommand\OSp{\mathrm{OSp}}

\newcommand\Map{\mathrm{Map}}

\newcommand{\bM}{\overline{\cM}}

\newcommand{\cM}{\mathcal{M}}
\newcommand{\cF}{\mathcal{F}}
\newcommand{\cB}{\mathcal{B}}
\newcommand{\cN}{\mathcal{N}}
\newcommand{\cH}{\mathcal{H}}
\newcommand{\fM}{\mathfrak{M}}
\newcommand{\fT}{\mathfrak{T}}
\newcommand{\fL}{\mathfrak{L}}
\newcommand{\bfM}{\overline{\mathfrak{M}}}
\newcommand{\bP}{\boldsymbol{P}}
\newcommand{\bsa}{\boldsymbol{a}}
\newcommand{\Tgn}{\Theta_{g,n}}

\title{The Super Virasoro Minimal String \\ from 3d Supergravity}

\author{Lorenz Eberhardt \orcidlink{0000-0003-1912-2211}}\emailAdd{l.eberhardt@uva.nl}

\affiliation{Institute for Theoretical Physics,
University of Amsterdam, Amsterdam, 1098XH, NL}

\abstract{The super Virasoro minimal string is defined by coupling spacelike and timelike super Liouville theories on the worldsheet. There are four different theories 0A$^\pm$ and 0B$^\pm$ depending on discrete choices on the worldsheet. We show that these theories arise naturally from quantization of 3d supergravity, and the amplitudes compute the dimension ($+$) or superdimension ($-$) of the space of $\cN=1$ superconformal blocks modulo crossing symmetry.
Both 0A$^+$ and 0B$^+$ are perturbatively dual to the same matrix integral as the bosonic Virasoro minimal string, while 0B$^-$ is perturbatively dual to a matrix integral with an inverse square root singularity. We show that all non-trivial perturbative amplitudes of the 0A$^-$ theory vanish.
}

\begin{document}

\maketitle

\makeatletter
\g@addto@macro\bfseries{\boldmath}
\makeatother
\section{Introduction}

Non-critical string theories are a particularly rich and fruitful arena to explore worldsheet string theory and its dualities. In works of the 1990s, it was realized that the $(p,q)$ minimal string exhibits a dual description in terms of a matrix integral \cite{Brezin:1990rb, Douglas:1989ve, Gross:1989vs} and the $c=1$ string in terms of matrix quantum mechanics \cite{Klebanov:1991qa, Ginsparg:1993is, Jevicki:1993qn, Polchinski:1994mb}. 
The original motivation for these dualities arises from considering multicritical limits of triangulated surfaces. Another route relates the minimal string to topological strings on a non-compact Calabi-Yau 3-fold \cite{Dijkgraaf:2002fc}. From a worldsheet perspective such dualities are typically obscured and amount to seemingly miraculous predictions for moduli space integrals. 
More recently, further examples of such correspondences were found such as the Virasoro minimal string \cite{Collier:2023cyw} and the complex Liouville string \cite{Collier:2024kmo}. 

Recent years have shown an improved handle on the worldsheet theories, which made it possible to directly derive parts of these dualities from the worldsheet \cite{Zamolodchikov:2003yb, Belavin:2006ex, Collier:2024kwt, Khromov:2025awh} and subject them to stringent numerical tests \cite{Belavin:2008kv,Balthazar:2017mxh, Collier:2023cyw, Collier:2024kwt}. In particular, the Virasoro minimal string and the complex Liouville string are especially tractable since one can exploit analyticity of their amplitudes.
\medskip 

The Virasoro minimal string (VMS) is in a sense the simplest theory among those mentioned above. Its worldsheet theory consists of spacelike Liouville theory at central charge $c$ coupled to timelike Liouville theory at central charge $26-c$. As such, its amplitudes depend continuously on $c$ and the external Liouville momenta $P_i$. It was noted in \cite{Collier:2023cyw} that this theory counts the number of crossing symmetric linear combinations of conformal blocks on the surface $\Sigma_{g,n}$. It can thus be viewed as a gravitational analogue of Verlinde's formula \cite{Verlinde:1988sn}.

This relation arises from quantization of 3d gravity with negative cosmological constant, whose phase space is Teichm\"uller space \cite{Collier:2023fwi}. The wavefunctions can be identified with conformal blocks \cite{Verlinde:1989ua, Teschner:2003em} and their inner product involves integration over Teichm\"uller space, exactly of the kind that appears in the VMS. The worldsheet integral computes the number of such conformal blocks which can be computed from an index theorem (Atiyah's $L^2$ index theorem \cite{Atiyah:1976vna}). This reduces the worldsheet integral to a cohomological expression on the moduli space $\bM_{g,n}$ for which well-established mathematical methods exist. In particular, such expressions are well-known to be reproduced by a dual matrix integral.
The density of states of this matrix integral is the chiral half of the Cardy density of states of 2d CFTs 
\be 
\rho(E)=\frac{2\pi \sinh(2\pi b \sqrt{E})\sinh(2\pi b^{-1} \sqrt{E})}{\sqrt{E}}\ , \label{eq:VMS density of states}
\ee
with $E=h-\frac{c-1}{24}$. It features the characteristic square root edge at $E=0$. Based on this density of states, one can apply the topological recursion procedure to compute any perturbative amplitude \cite{Eynard:2011ga}.
The VMS amplitudes are called the quantum volumes since they recover the Weil-Petersson volumes of moduli space in the $c \to \infty$ limit. Thus, the VMS constitutes a deformation of JT gravity. 

Similar cohomological formulas have been proposed for the amplitudes of the $c=1$ string \cite{Collier:2026pxi}, the $(2,p)$ minimal string \cite{Artemev:2024rck} and the complex Liouville string \cite{Collier:2024kmo}.
In each of these cases, their amplitudes admit an expansion in Feynman diagrams in which the quantum volumes serve as the vertex factors.

Various supersymmetric generalizations of these theories and their dualities have been studied, such as for the $(p,q)$ minimal string \cite{Seiberg:2003nm, Klebanov:2003wg, Maldacena:2005he} and the $c=1$ string \cite{Takayanagi:2003sm, Douglas:2003up, Balthazar:2022atu}. Surprisingly, the supersymmetric analogue of these theories is often closely related to their bosonic avatars, a fact that is also obscured from the worldsheet theory. 
\medskip

In this paper, we consider the super Virasoro minimal string from a 3d gravity perspective. Some aspects of the theory have already been studied from the worldsheet \cite{Muhlmann:2025ngz, Rangamani:2025wfa} and the matrix integral perspective \cite{Johnson:2024fkm, Johnson:2025vyz}. 
On the worldsheet, the super VMS is defined by coupling $\cN=1$ spacelike Liouville theory of central charge $c=\frac{3}{2}+3(b+b^{-1})^2$ to $\cN=1$ timelike Liouville theory of central charge $15-c$.

The super Virasoro minimal string can also be obtained by studying the quantization of $\cN=1$ supergravity in three dimensions. As a consequence, the worldsheet amplitudes count the number of crossing symmetric combinations of $\cN=1$ superconformal blocks on the (super) surface $\Sigma_{g,n}$. There are two discrete signs in this counting. For only NS-sector punctures, the spin structure on $\Sigma_{g,n}$ can be even or odd and one can choose to weigh the odd spin structures by a sign $(-1)^\zeta$ in the counting. The unweighted count corresponds to type 0B string theory which does feature both NS- and R-sector vertex operators, while the weighted count corresponds to type 0A string theory, which only has NS-sector vertex operators.\footnote{This convention is standard in string perturbation theory, see e.g.\ \cite{Klebanov:2003wg, Kaidi:2019tyf}, but opposite to the convention employed in super JT gravity \cite{Stanford:2019vob} and in a previous treatment of the super Virasoro minimal string \cite{Johnson:2024fkm, Johnson:2025vyz}. The origin of this disagreement was discussed in \cite[Appendix A.4]{Stanford:2019vob}. Since our discussion centers around string perturbation theory, we follow that convention. \label{footnote:0A/0B convention}}

Additionally, we can choose to include a sign $(-1)^\text{F}$ into the counting. We denote the theories without $(-1)^\text{F}$ by 0A$^+$/0B$^+$ and those with $(-1)^\text{F}$ counting by 0A$^-$/0B$^-$. The appearance of the same sign was already observed in the worldsheet theory in \cite{Rangamani:2025wfa, Muhlmann:2025ngz}, where it corresponds to a sign flip of a chiral supercharge in one of the Liouville theories.

One can count the number of superconformal blocks in terms of an index theorem on super moduli space. We show that the resulting expression can be reduced to its bosonic subspace, the moduli space of spin curves. Standard theorems in algebraic geometry \cite{Chiodo:2008} allow us to write down explicit intersection number expressions for the super quantum volumes. We further carry out the sum over spin structures and reduce everything to standard intersection numbers on the moduli space of curves. For the 0A/B$^+$ theories, we obtain compact formulas in terms of the basic $\kappa$ and $\psi$-classes on $\bM_{g,n}$, while the 0A/B$^-$ cases additionally lead to the appearance of the so-called Theta class (or a variation thereof) \cite{Norbury:2017eih}.
From these expressions, we show that all these theories can be described via a matrix integral. More precisely, we show that
\begin{enumerate}
    \item Type 0A$^+$ is perturbatively \emph{equivalent} to bosonic VMS and shares the same topological recursion.
    \item Type 0B$^+$ is also perturbatively equivalent to bosonic VMS. Moreover, amplitudes involving Ramond-sector vertex operators are identical (up to normalization) to amplitudes with only NS-sector vertex operators.
    \item Apart from the sphere two-point function, all perturbative amplitudes of type 0A$^-$ vanish identically. This vanishing is non-trivial. We prove the vanishing of the corresponding intersection numbers by employing cohomological field theory techniques in appendix~\ref{app:arf theta class}. As far as we are aware, this vanishing theorem has not appeared in the mathematical literature.
    \item Type 0B$^-$ with NS-sector vertex operators is dual to a matrix integral with density of states
\be 
\rho(E)=\frac{2\pi \cosh(2\pi b \sqrt{E})\cosh(2\pi b^{-1} \sqrt{E})}{\sqrt{E}}\ , \label{eq:super VMS density of states cosh}
\ee
which contrary to \eqref{eq:VMS density of states} features an inverse square root singularity as $E \to 0$. As a consequence this model behaves quite differently to the 0A/B$^+$ theories. It is governed by the Br\'ezin-Gross-Witten $\tau$-function \cite{Brezin:1980rk, Gross:1980he,Chidambaram:2022cqc, Norbury:2020vyi} and all its tree-level amplitudes vanish (but higher genus amplitudes are non-zero).
\end{enumerate}
\medskip
This paper is organized as follows. In section~\ref{sec:quantization}, we discuss 3d supergravity, its quantization and relation to the super VMS and the cohomological reduction via the index theorem. We also discuss the origin of the four different theories. In section~\ref{sec:intersection theory}, we further analyze the resulting cohomological expressions and explain how they can be reduced to intersection numbers on ordinary moduli space, which ultimately allows us to relate them to a matrix integral. We end with a short discussion in section~\ref{sec:discussion}. In appendix~\ref{app:hard topological recursion}, we review topological recursion for density of states with inverse square root singularity and its relation to intersection theory. In appendix~\ref{app:arf theta class}, we prove the vanishing of the intersection numbers of the type 0A$^-$ theory.

\section{Super VMS amplitudes from 3d supergravity} \label{sec:quantization}
We will study chiral 3d $\cN=(1,1)$ supergravity with negative cosmological constant in canonical quantization. Our goal in this section is to give a description of the Hilbert space on a spatial Cauchy slice. The logic is in parallel to the construction of bosonic Virasoro TQFT \cite{Collier:2023fwi}. Some aspects of the quantization were also worked out in \cite{Bhattacharyya:2024vnw}.
\subsection{Chern-Simons theory and phase space}
\paragraph{Phase space from Chern-Simons theory.} At the level of the action and the equations of motion $\cN=(1,1)$ supergravity is equivalent to Chern-Simons theory with gauge group $(\OSp(1|2)/\ZZ_2) \times (\OSp(1|2)/\ZZ_2)$ \cite{Witten:1988hc}. Since the gauge group factorizes, we can discuss quantization of one factor separately. The two factors should be thought of as holomorphic and antiholomorphic states. We refer to the theory with a single $\OSp(1|2)/\ZZ_2$ factor as chiral $\cN=1$ supergravity. 

There are various global issues in this identification. We picked the particular global form $\OSp(1|2)/\ZZ_2$ of the gauge group. In the action, only the algebra $\mathfrak{osp}(1|2)$ is visible. The bosonic subgroup of $\OSp(1|2)$ is $\mathrm{O}(1) \times \Sp(2,\RR) \cong \ZZ_2 \times \SL(2,\RR)$. The $\ZZ_2$ factor acts by a minus sign on the fermionic directions. We thus think of it as $(-1)^\text{F}$. Bosons transform in $\SL(2,\RR)$ representations that originate from the quotient $\mathrm{SO}(1,2) \cong \PSL(2,\RR) \cong \SL(2,\RR)/\ZZ_2$, while fermions transform in spinorial representations of the local Lorentz group $\mathrm{SO}(1,2)$. This means that the correct bosonic subgroup that encapsulates all the representations correctly is $((-1)^\text{F} \times \SL(2,\RR))/\ZZ_2$. This explains our choice of global form $\OSp(1|2)/\ZZ_2$.

The phase space of Chern-Simons theory is well-known. Fix a Cauchy slice $\Sigma_{g,n}$. Each puncture carries a label specifying Neveu-Schwarz (antiperiodic boundary conditions for the fermions) or Ramond type (periodic boundary conditions), together with a group-valued monodromy.

Naively, an initial condition amounts to choosing an $\OSp(1|2)/\ZZ_2$ gauge field on $\Sigma_{g,n}$. The gauge symmetry of the Chern-Simons theory imposes the flatness constraint on the gauge bundle on the initial surface and identifies gauge equivalent bundles. The constrained phase space of the Chern-Simons theory is therefore the space of flat $\OSp(1|2)/\ZZ_2$ bundles\footnote{We use gothic letters for supermanifolds.}
\be 
\fM_{\OSp(1|2)/\ZZ_2} \cong \{ \rho:\pi_1(\Sigma_{g,n}) \longrightarrow \OSp(1|2)/\ZZ_2 \}/(\OSp(1|2)/\ZZ_2)\ , \label{eq:MOSp12}
\ee
i.e.\ the space of all homomorphisms of the fundamental group given by the monodromy representation.

The phase space \eqref{eq:MOSp12} is disconnected. Only some of the connected components represent regular metrics and gravitinos on the Cauchy surface. In the bosonic case, it is known that one can consistently restrict the phase space to such a regular component known as Teichm\"uller space. Teichm\"uller space describes hyperbolic surfaces (together with a choice of $a$ and $b$-cycles) where the flat $\PSL(2,\RR)$ bundle is realized as the tangent bundle of the surface. In the supersymmetric case, the situation is analogous: one can pick out the super Teichm\"uller components of this phase space. We will discuss their structure further below.

Thus we conclude that the phase space is super Teichm\"uller space $\fT_{g,\bsa}$. Here, $\bsa=(a_1,\dots,a_n)$ labels the R- and NS-sectors with the convention
\be 
a_i= \begin{cases}
    1 \ , \quad &\text{NS-puncture}\ , \\
    0 \ , \quad &\text{R-puncture}\ .
\end{cases} \label{eq:a NS-R sector dictionary}
\ee
We will discuss its geometry further below.
\paragraph{K\"ahler quantization.} Super Teichm\"uller space has not only a description in terms of moduli of super hyperbolic surfaces, but also as moduli of super \emph{Riemann} surfaces \cite{Crane:1986uf}. For the purposes of the quantization, it is useful to switch to the complex viewpoint. 
Quantization of ordinary Teichm\"uller space is a well-developed topic in the mathematics literature \cite{Kashaev:1998fc, Chekhov:1999tn, Teschner:2003em} with recent works establishing extensions to super Teichm\"uller space \cite{Aghaei:2020otq}. For our purposes, we will follow the route described in \cite{Verlinde:1989ua} for the quantization of ordinary Teichm\"uller space. In that paper, a convenient choice of polarization for the quantization in the bosonic case was made inspired by the AdS boundary conditions. One can solve the flatness constraints of the Chern-Simons quantization at the quantum level. Ultimately, this reduces the wavefunction to a holomorphic function on Teichm\"uller space that satisfies the Virasoro Ward identities, i.e.\ carries the same conformal anomaly as a conformal block of central charge $c$ and conformal weights $h_i$ associated to the punctures. The central charge $c$ is the Brown-Henneaux central charge \cite{Brown:1986nw} $c=\frac{3\ell_\mathrm{AdS}}{2G_\text{N}}$, up to scheme dependent higher loop corrections. The conformal weight is related to the monodromy of the gauge field around the puncture. When parametrizing the conformal weight via the Liouville momentum as $h=\frac{c-1}{24}+P^2$, the relation is semiclassically $P \sim \frac{L}{4\pi b}$ \cite{Teschner:2003at}, with $L$ the geodesic length. 

In the supersymmetric case, an analogous procedure should carry through: the $\cN=(1,1)$ Chern-Simons formulation admits the same polarization choice, and the flatness constraints reduce to the $\cN=1$ super Virasoro Ward identities. The punctures are again labelled by the conformal weight data as well as the label $a \in \{0,1\}$ specifying NS- or R- sector as in \eqref{eq:a NS-R sector dictionary}. The Hilbert space is therefore spanned by $\cN=1$ superconformal blocks of central charge $c$. This is clear from the perspective of quantizing the constrained phase space. Superconformal blocks are holomorphic objects on super Teichm\"uller space (moduli of the super Riemann surface up to crossing transformations) specified by the central charge, the conformal weights and the external NS/R-sector labels. The internal conformal weights and NS/R-sector labels should be associated to the index labelling the basis. Therefore, it will not be necessary to repeat the analysis of \cite{Verlinde:1989ua} in the supersymmetric setting for our purposes. 

This expectation on the relation of the wavefunctions to superconformal blocks is also borne out of the AdS$_3$/CFT$_2$ correspondence. The path integral of Euclidean $\cN=(1,1)$ supergravity on a 3d manifold $M$ produces a wavefunction on the asymptotic AdS boundary $\partial M$. From AdS/CFT, we expect that the wavefunction carries the anomalies of a 2d CFT, see \cite{Henningson:1998gx, Coussaert:1995zp}. Since we are only analyzing a chiral half of 3d supergravity, it should follow that the space of wavefunctions is indeed spanned by superconformal blocks as detailed above.

Going forward, we will employ the following standard parametrization of the central charge and conformal weights inspired from $\cN=1$ Liouville theory
\be 
c=\frac{3}{2}+3Q^2\ , \qquad Q=b+b^{-1}\ , \qquad h=\frac{Q^2}{8}+\frac{1}{16} \, \delta_{a,0}+\frac{P^2}{2}\ . \label{eq:c h parametrization}
\ee
The shift of $\frac{1}{16}$ in the Ramond sector is the familiar Ramond-sector ground state energy.

\paragraph{Super Teichm\"uller space.} Let us explain some important details of super Teichm\"uller space. It is a supermanifold of dimension
\be 
\dim_\CC \fT_{g,\bsa}= 3g-3+n \mid 2g-2+\tfrac{1}{2}(|\bsa|+n)\ .
\ee
Here, we separate the bosonic and fermionic dimension by $\mid$. We also defined $|\bsa|= \sum_i a_i$. In other words, NS-punctures contribute $1 \mid 1$ to the dimension, while R-punctures contribute $1 \mid \frac{1}{2}$ to the dimension. In particular, the number of R-punctures has to be even for super Teichm\"uller space to exist. 

The bosonic subspace of super Teichm\"uller space is the Teichm\"uller component of the moduli space of flat $\SL(2,\RR)$ bundles (not flat $\PSL(2,\RR)$ bundles). This space was studied in detail by Goldman \cite{Goldman:1988} who showed that the Teichm\"uller component has $2^{2g}$ connected components corresponding to the lift of the $\PSL(2,\RR)$ generators associated to the $2g$ cycles of the surface to the double cover. They correspond to the $2^{2g}$ choices of spin structures on the surface.

\subsection{The Hilbert space and super Virasoro TQFT}
\paragraph{Generalities on the inner product.} In K\"ahler quantization, one constructs a holomorphic line bundle $\mathcal{L}$ whose first Chern class agrees with the symplectic form on phase space. Wavefunctions $\Psi$ are holomorphic sections of this line bundle.
One can write down an inner product on the Hilbert space involving the K\"ahler potential $\mathcal{K}$. It was shown in \cite{Verlinde:1989ua} and later refined in \cite{Collier:2023fwi} that the appropriate choice for $\mathrm{e}^{-\mathcal{K}}$ in bosonic 3d gravity is timelike Liouville theory. To motivate this, we need to recall the role of the K\"ahler potential in the inner product, which schematically takes the following form in K\"ahler quantization
\be
\langle \Psi \mid \Psi' \rangle=\int_{\cM_\text{phase}} \d \mu_\text{ghosts}\, \Psi^* \Psi'\, \mathrm{e}^{-\mathcal{K}}\ ,
\ee
where $\d \mu_\text{ghosts}$ is the appropriate ghost measure on the phase space $\cM_\text{phase}$.
The choice of the K\"ahler potential has to achieve two things: it has to render the phase space integral well-defined, i.e.\ absorb the cocycle factors from the wavefunctions, and has to realize a possible group action by a group $\Gamma$ on $\cM_\text{phase}$ in terms of unitary operators on the Hilbert space. This implies that $\mathcal{K}$ has to be $\Gamma$-invariant.

In our case, $\cM_\text{phase}$ is super Teichm\"uller space and we will denote the line bundle by $\fL_{\bP,\bsa}^{(b)}$. It captures the conformal anomaly (labelled by $b$), the weight of the punctures (labelled by $\bP$) as well as whether they are NS- or R-punctures (labelled by $\bsa$) \cite{Friedan:1986ua}. $\Gamma$ will be the group of crossing transformations (the mapping class group).
The discussion of integrals over (super) Teichm\"uller space is in parallel with (super)string theory integrals, where one integrates over moduli space \cite{Witten:2012bg}. The two spaces differ only globally. 

\paragraph{Explicit form of the inner product.}
Let us first recall the bosonic case. There, the exponentiated K\"ahler potential $\mathrm{e}^{-\mathcal{K}}$ has to carry central charge $\hat{c}=26-c$ and conformal weights $\hat{h}_i=1-h_i$. The condition of unitarity of crossing transformations implies that $\mathrm{e}^{-\mathcal{K}}$ has to transform as a crossing symmetric 2d CFT partition function. It also has to be a continuous function of $b$ and $h_i$. In the bosonic case, it was argued in \cite{Collier:2023fwi} that the only candidate theory for $\mathrm{e}^{-\mathcal{K}}$ that meets these criteria is timelike Liouville theory.

The supersymmetric case is analogous. In this case, $\mathrm{e}^{-\mathcal{K}}$ has to carry central charge $\hat{c}=15-c$ and conformal weights $\hat{h}_i=\frac{1}{2}+\frac{1}{8}\delta_{a,0}-h_i$. Here, $15=26-11$ is the ghost central charge from the $\mathfrak{b}\mathfrak{c}$ and $\beta\gamma$ ghosts of string theory, while $\frac{1}{2}+\frac{1}{8}\delta_{a,0}$ appears in the mass-shell condition. The superconformal ghosts provide the ground state energy $-1+(\frac{1}{2}-\frac{1}{8}\delta_{a,0})$. Indeed, to construct an unintegrated bosonic vertex operator, we multiply the worldsheet CFT vertex operator by $\mathfrak{c} \mathrm{e}^{-\varphi}$ in the NS-sector and by $\mathfrak{c} \mathrm{e}^{-\frac{1}{2} \varphi}$ in the R-sector (suppressing right-movers). Here $\varphi$ is the bosonized superconformal ghost. We have $h(\mathfrak{c})=-1$ and $h(\mathrm{e}^{-\varphi})=\frac{1}{2}$ and $h(\mathrm{e}^{-\frac{1}{2}\varphi})=\frac{3}{8}$.
In view of the parametrization \eqref{eq:c h parametrization} this corresponds to $\hat{b}=i b$ and $\hat{P}=i P$ in both the NS- and R-case. 
The integral over super Teichm\"uller space is locally the same as in superstring theory.
In \cite{Rangamani:2025wfa,Muhlmann:2025ngz} it was shown that $\cN=1$ timelike Liouville theory is the unique solution of the analytic bootstrap with $\hat{c} \le \frac{3}{2}$ when one assumes a diagonal scalar spectrum without degeneracy and analytic dependence of the CFT data on the external parameters. The three-point functions are completely determined in these references. 
These observations motivate the following inner product between two wavefunctions $\Psi$ and $\Psi'$,
\be 
\langle \Psi \mid \Psi' \rangle := \int_{\fT_{g,\bsa}} \d \mu_{\mathfrak{b}\mathfrak{c},\beta\gamma} \, \Psi^* \Psi' \, \langle \hat{V}_{\hat{P}_1}^{a_1} \cdots \hat{V}_{\hat{P}_n}^{a_n} \rangle_g\ . \label{eq:inner product}
\ee
Here, the measure $\d \mu_{\mathfrak{b}\mathfrak{c},\beta\gamma}$ is the superstring measure obtained from the $\mathfrak{b}\mathfrak{c}$ and $\beta\gamma$ ghost correlation functions. We write $\hat{V}_{\hat{P}_i}^{a_i}$ for the $\cN=1$ timelike Liouville vertex operators in the sector labelled by $a_i$.

\paragraph{Orthogonal basis and spacelike super Liouville.} We will now argue that superconformal blocks with heavy internal weights form a complete basis with respect to this inner product. Let us first analyze normalizability. For concreteness, consider the degeneration limit where two points collide on the surface. Let $q$ be a local plumbing parameter describing the degeneration. The local behaviour of the integrand is
\be 
\frac{\big|q^{-1-1+\frac{1}{24}(c-\frac{3}{2}\delta_{a,1})+\frac{1}{2}P^2+\frac{1}{24}(\hat{c}-\frac{3}{2}\delta_{a,1})+(\frac{1}{2}-\frac{1}{8} \delta_{a,0})}\big|^2}{|\log |q||^{1/2}} = \frac{|q^{-1+\frac{1}{2}P^2}|^2}{|\log |q||^{1/2}}\ .
\ee
Here we used that integrated vertex operators have conformal weight $(1,1)$. The internal weight appearing in the OPE is $\frac{1}{24}(c-\frac{3}{2}\delta_{a,1})+\frac{1}{2}P^2$ for the conformal blocks, while the timelike Liouville contribution is bounded from below by $\frac{1}{24}(\hat{c}-\frac{3}{2}\delta_{a,1})$. For small plumbing parameter, the integral over the internal Liouville momenta can be performed via saddle point approximation at $P=0$. The Gaussian integral produces the factor $|\log |q||^{-1/2}$ as in the bosonic case, see \cite{Ribault:2015sxa}.
The contribution $\frac{1}{2}-\frac{1}{8} \delta_{a,0}$ are the contributions of the $\beta\gamma$ ghosts to the internal conformal weight as discussed above. Thus, the integral over the vicinity of $q=0$ converges provided that $P \in \RR$. This indicates precisely that the Hilbert space contains the superconformal blocks appearing in super Liouville theory. 

In fact, we claim that superconformal blocks form a $\delta$-function normalizable orthogonal basis to this Hilbert space. We adapt an argument from the bosonic construction \cite{Collier:2023fwi}. Let $\{|e_i \rangle\}_{i\ge 1}$ be a countable orthonormal basis for this Hilbert space. Then we can form the combination $B(\mathfrak{m},\tilde{\mathfrak{m}})= \sum_{i\ge 1} |e_i(\mathfrak{m}) \rangle | e_i(\tilde{\mathfrak{m}})\rangle^*$. Here $\mathfrak{m} \in \fT_{g,\bsa}$ stands collectively for all supermoduli. $B$ itself behaves as a CFT correlation function of central charge $c$ and conformal weights $h_i$, where $\mathfrak{m}$ corresponds to left-moving moduli and $\tilde{\mathfrak{m}}$ to right-moving moduli. Crossing symmetry follows because of unitarity of crossing transformations. $B$ is known as the Bergman kernel of the Hilbert space and existence as well as uniqueness can also be established by invoking the Riesz representation theorem. The main point is again that there is a unique such object --- namely spacelike $\cN=1$ Liouville theory \cite{Rashkov:1996np, Poghossian:1996agj, Belavin:2007gz}. In other words, $B=\langle V_{P_1}^{a_1} \cdots V_{P_n}^{a_n} \rangle_g$ must be the correlation function of $\cN=1$ Liouville theory. From this, orthonormality follows immediately because the spectrum of $\cN=1$ Liouville theory is diagonal. 

\paragraph{Super Virasoro TQFT.} At this point, we have defined a Hilbert space $\cH_{g,\bP,\bsa}^{(b)}$ associated to the data of $g$, $\bP$ and $\bsa$. Crossing transformations act unitarily on this Hilbert space via a set of crossing matrices. They are explicitly known for $\cN=1$ superconformal blocks \cite{Hadasz:2007wi, Hadasz:2013bwa, Apresyan:2023vxl} and generalize the bosonic Ponsot-Teschner fusion kernel \cite{Ponsot:1999uf, Ponsot:2000mt}. They form a \emph{projective} unitary representation. This provides the data to define a 3d TQFT, which we call super Virasoro TQFT. It gives an all-loop definition of 3d $\cN=1$ supergravity in parallel to the bosonic construction developed in \cite{Collier:2023fwi}. We will not develop this theory further here.

\subsection{Super Virasoro minimal string}
\paragraph{The mapping class group.} In 3d gravity, one has to gauge the mapping class group $\Gamma=\Map_{g,n}$, i.e.\ the set of large diffeomorphisms. The mapping class group in the supersymmetric setting is unchanged, since the fermionic directions of super Teichm\"uller space are contractible and therefore do not support large diffeomorphisms. The mapping class group acts non-trivially on the phase space and needs to be gauged. One can perform this gauging either before or after quantization. 
In the approach where one gauges first, one reduces the phase space to the quotient space $\fM_{g,\bsa} = \fT_{g,\bsa}/\Map_{g,n}$. This is the moduli space of super Riemann surfaces. It has finite (super) volume and thus one expects a finite dimensional Hilbert space.
However, there is a subtlety. On $\fT_{g,\bsa}$, wavefunctions are sections of a holomorphic line bundle $\fL_{\bP,\bsa}^{(b)}$. 
Unless the projective phases cancel (which imposes quantization conditions on the central charge and the conformal weights), the line bundle does not descend to a proper holomorphic line bundle on $\fM_{g,\bsa}$. 

\paragraph{The dimension.} Nonetheless, it is possible to associate a measure of the dimension of the gauged Hilbert space. The right framework for this is Atiyah's analytic $L^2$-index as well as subsequent generalizations \cite{Atiyah:1976vna, Moscovici:1982fv, Mathai:2006fai}. The basic definition of a reasonable dimension for the space of sections is
\begin{align}
\dim_\Gamma \cH_{g,\bP,\bsa}^{(b)}&=\int_{\fM_{g,\bsa}}\d \mu_{\mathfrak{b}\mathfrak{c},\beta\gamma}(\mathfrak{m})\, B(\mathfrak{m},\mathfrak{m})\, \langle \hat{V}_{\hat{P}_1}^{a_1} \cdots \hat{V}_{\hat{P}_n}^{a_n} \rangle(\mathfrak{m}) \\
&=\int_{\fM_{g,\bsa}}\d \mu_{\mathfrak{b}\mathfrak{c},\beta\gamma}\, \langle V_{P_1}^{a_1} \cdots V_{P_n}^{a_n} \rangle \, \langle \hat{V}_{\hat{P}_1}^{a_1} \cdots \hat{V}_{\hat{P}_n}^{a_n} \rangle\ . \label{eq:super Virasoro dimension}
\end{align}
Here, $B(\mathfrak{m},\mathfrak{m})$ is the Bergman kernel defined above evaluated on the diagonal. As explained above, it equals the $\cN=1$ super Liouville correlation function. The resulting integral looks exactly like a worldsheet integral of type 0 superstring theory. This defines the amplitudes of the super Virasoro minimal string \cite{Collier:2023cyw, Muhlmann:2025ngz, Rangamani:2025wfa}.

Intuitively, \eqref{eq:super Virasoro dimension} provides a reasonable definition of the dimension of the Hilbert space after gauging. If we naively use the definition $B(\mathfrak{m},\tilde{\mathfrak{m}})= \sum_{i\ge 1} |e_i(\mathfrak{m}) \rangle | e_i(\tilde{\mathfrak{m}})\rangle^*$ and recall that $\fM_{g,\bsa} \cong \fT_{g,\bsa}/\Gamma$, we obtain
\be 
\dim_\Gamma \cH_{g,\bP,\bsa}^{(b)} \sim \frac{1}{|\Gamma|} \sum_{i \ge 1} \langle e_i \mid e_i \rangle =\frac{1}{|\Gamma|} \dim \cH_{g,\bP,\bsa}^{(b)}\ , \label{eq:dim Gamma}
\ee
as expected for such a gauging (even though the ratio of the factors on the RHS is of course ill-defined). Thus, we will take \eqref{eq:super Virasoro dimension} to be the natural definition of the Hilbert space dimension after gauging.

\paragraph{Interpretation as an index.} Conceptually, computing dimensions of spaces of sections of holomorphic line bundles is hard. It is much simpler to consider an index. The Hilbert space consists of holomorphic sections of a line bundle. A simpler quantity is the Euler characteristic, given by the alternating sum 
\be 
\chi_\Gamma(\fT_{g,\bsa},\fL_{\bP,\bsa}^{(b)}) = \sum_{k\ge 0} (-1)^k \dim_\Gamma H^k(\fT_{g,\bsa},\fL_{\bP,\bsa}^{(b)})\ , \label{eq:Euler characteristic}
\ee
since it can be computed with the help of index theorems and therefore does not jump under small deformations of the external data. In fact, we shall now argue that $\dim_\Gamma H^k(\fT_{g,\bsa},\fL_{\bP,\bsa}^{(b)})=0$ for $k \ge 1$ and therefore passing to the index doesn't actually change the dimension.

First notice that since $\cN=1$ timelike Liouville theory can be defined for $\hat{c} \le \frac{3}{2}$, \eqref{eq:super Virasoro dimension} is defined for $c \ge \frac{27}{2}$, i.e.\ $b \in \RR$. 
In K\"ahler quantization, the line bundle $\fL_{\bP,\bsa}^{(b)}$ is \emph{positive} since its first Chern class is represented by the symplectic form $\omega$ and $\omega(\bullet,J \bullet)$ is the K\"ahler metric and positive. Kodaira's vanishing theorem tells us that $(\fL_{\bP,\bsa}^{(b)})^{\otimes n}$ for high enough power of $n$ has the desired vanishing property. However, taking the $n$-th tensor power of the line bundle just corresponds to rescaling the conformal weight and the central charge by a factor of $n$, see also \eqref{eq:c_1 L formula}. Therefore, vanishing is true for large enough $b$ and $P_i$. However, both the index and the worldsheet integral are analytic functions in $b$ and $P_i$. Therefore, agreement of the worldsheet integral with the index follows for all $P_i \in \RR$ and $b \in \RR$.

This argument requires some care because $\fT_{g,\bsa}$ is non-compact and Kodaira's vanishing theorem only holds for compact manifolds. One possible way to get around this is by considering the quotient $\fM_{g,\bsa}=\fT_{g,\bsa}/\Gamma$. It is not compact, but admits a natural compactification $\overline{\fM}_{g,\bsa}$ (the Deligne-Mumford compactification). $\smash{\fL_{\bP,\bsa}^{(b)}}$ does not necessarily descend to a line bundle on $\fM_{g,\bsa}$, since this requires certain quantization conditions. This does not modify the argument and we can momentarily assume this to be the case. We can then extend $\smash{\fL_{\bP,\bsa}^{(b)}}$ and the tangent bundle to the Deligne-Mumford compactification $\smash{\bfM_{g,\bsa}}$ with line bundle $\smash{\overline{\fL}_{\bP,\bsa}^{(b)}}$ and apply Kodaira vanishing on the compact space directly. For the identification of $\smash{H^0(\bfM_{g,\bsa},\overline{\fL}_{\bP,\bsa}^{(b)})}$ with $\dim_\Gamma \cH$, we note that sections on $\bfM_{g,\bsa}$ are automatically $L^2$ on the interior (since the Weil-Petersson volume is finite), while conversely, the plumbing analysis above shows that $L^2$ holomorphic sections are bounded near the compactification divisor and therefore extend holomorphically across it by the Riemann extension theorem. Ampleness of the extended line bundle on the compactification follows because $\smash{c_1(\fL_{\bP,\bsa}^{(b)})}$ is a positive linear combination of $\kappa_1$ and the $\psi$-classes (see \eqref{eq:c_1 L formula} below), which generate an ample class on $\smash{\overline{\cM}_{g,n}}$ for large enough coefficients \cite{Cornalba:1988ens}.
Therefore
\be 
\dim_\Gamma \cH_{g,\bP,\bsa}^{(b)}=\chi_\Gamma(\fT_{g,\bsa},\fL_{\bP,\bsa}^{(b)})=\int_{\fM_{g,\bsa}}\d \mu_{\mathfrak{b}\mathfrak{c},\beta\gamma}\, \langle V_{P_1}^{a_1} \cdots V_{P_n}^{a_n} \rangle \, \langle \hat{V}_{\hat{P}_1}^{a_1} \cdots \hat{V}_{\hat{P}_n}^{a_n} \rangle\ . \label{eq:index delbar worldsheet integral}
\ee
\paragraph{The index theorem.} Since both the line bundle and the tangent bundle extend to the compactification, the Hirzebruch-Riemann-Roch theorem on the compact space $\bfM_{g,\bsa}$ gives the standard form of the index theorem,
\be
\dim_\Gamma \chi_\Gamma(\fT_{g,\bsa},\fL_{\bP,\bsa}^{(b)})=\int_{\bfM_{g,\bsa}} \std(\fM_{g,\bsa})\, \mathrm{e}^{c_1(\overline{\fL}_{\bP,\bsa}^{(b)})}\ . \label{eq:naive index theorem}
\ee
Here, $\std(\fM_{g,\bsa})$ denotes the super Todd class of the tangent bundle of super moduli space. For now this is a place holder. We discuss further below what this exactly means. $\smash{c_1(\overline{\fL}_{\bP,\bsa}^{(b)})}$ is the first Chern class of the line bundle $\smash{\overline{\fL}_{\bP,\bsa}^{(b)}}$. Since $\smash{\overline{\fL}_{\bP,\bsa}^{(b)}}$ is not necessarily a proper line bundle on $\smash{\overline{\fM}_{g,
\bsa}}$, $\smash{c_1(\overline{\fL}_{\bP,\bsa}^{(b)})}$ does not need to take values in integer cohomology.

Let us comment on some additional subtleties. As mentioned above, the line bundle on $\fT_{g,\bsa}$ does not necessarily lead to a proper line bundle after gauging because the action of crossing transformations on $\fL_{g,\bsa}$ is only projective. The projective case of Atiyah's index theorem was worked out in \cite{Mathai:2006fai}. Orbifold singularities are not a problem because one can pass to a finite cover $\widetilde{\fM}_{g,\bsa}$ that is free from orbifold singularities \cite{Boggi:2000aa}. The fact that we are working on supermoduli space is also not a problem. As we shall discuss below, one can reduce everything to the bosonic moduli space of spin curves. Thus, the right hand side of \eqref{eq:naive index theorem} should be viewed as a stand in for an expression on the bosonic submanifold that we discuss below. Essentially, the fermionic directions only contribute a finite amount of freedom, do not carry topology and do not modify the theorem substantially.

\subsection{Discrete symmetries} 
The discussion so far suggests that the amplitudes of the super Virasoro minimal string can be evaluated through intersection theory. As was discussed in \cite{Muhlmann:2025ngz, Rangamani:2025wfa}, there are actually four different Virasoro minimal strings, that we will denote by 0A$^+$, 0A$^-$, 0B$^+$ and 0B$^-$. 

\paragraph{0A vs.\ 0B theory and the Arf invariant.} We begin by recalling the difference of type 0A and 0B string theory. The super moduli space $\fM_{g,\boldsymbol{1}}$ (i.e.\ only NS-punctures) has \emph{two} connected components, while the super moduli space with at least two Ramond punctures only has one connected component. To understand why, it is sufficient to consider the moduli space of spin curves which is the bosonic subspace of $\fM_{g,\bsa}$. A (generalized) spin structure is a choice of line bundle $S$ such that $S^2 \cong \omega(-D)$, where $D=\sum_i (a_i-1) [z_i]$ is the divisor of Ramond punctures, meaning that a spinor has a square root branch cut at an R-puncture and no branch cut at an NS-puncture. Here, $\omega$ is the canonical line bundle on the surface whose sections are holomorphic 1-forms. For a spin structure with only NS-punctures, $\zeta=\dim H^0(S) \bmod 2$ is the mod-2 index and is a topological invariant \cite{Atiyah:1971rm}, also called the Arf invariant. Therefore, $(-1)^\zeta$ defines a topological field theory in two dimensions. It is not possible to define $(-1)^\zeta$ for a surface with Ramond punctures and consequently $\fM_{g,\bsa}$ in their presence does not decompose into two components.

The inclusion of the additional sign $(-1)^\zeta$ leads to the type 0A theory. Since this is only possible when no Ramond punctures are present, the type 0A theory only has NS-sector vertex operators. This can also be seen more directly from the worldsheet as a consequence of the sum over spin structures that project out the Ramond sector vertex operators for the type 0A theory. The type 0B theory does not include this sign.

\paragraph{$(-1)^\text{F}$.} There is another discrete sign one can introduce for both the type 0A and 0B theory. From the point of view of 3d gravity, it is obvious: since the theory has both bosonic and fermionic states, we can compute the dimension of the Hilbert space with or without the inclusion of $(-1)^\text{F}$. In the former, we are computing the superdimension while in the latter we are computing the dimension. It is also simple to introduce this sign on the intersection side. There are actually two different notions of super Todd class one can define which we denote by $\std_\pm(\fM_{g,\bsa})$. They are precisely defined below, see eqs.~\eqref{eq:tdB ch+ F combination} and \eqref{eq:tdB ch- F combination}.

On the worldsheet side, we need to weigh superconformal blocks with an odd number of fermions by a minus sign. Recall that the spacelike Liouville factor in \eqref{eq:index delbar worldsheet integral} arises from the Bergman kernel. It constitutes a particular linear combination of left- and right-moving conformal blocks, and we can insert $(-1)^\text{F}$ by flipping the sign of all left-moving fermions.\footnote{Since left- and right-moving blocks are paired diagonally, we can equivalently flip the sign of all right-moving fermions.} This is equivalent to flipping the sign of the left-moving supercharge $G^\text{sL}$ of spacelike Liouville. In other words, insertion of the factor $(-1)^\text{F}$ implies that a \emph{different} worldsheet $\cN=(1,1)$ super Virasoro algebra is gauged, namely $G=-G^\text{sL}+G^\text{tL}$ and $\tilde{G}=\tilde{G}^\text{sL}+\tilde{G}^\text{tL}$. The appearance of this sign possibility was noticed from the worldsheet in \cite{Muhlmann:2025ngz, Rangamani:2025wfa}. 

\paragraph{Summary.} To summarize the discussion, we have derived that
\be 
\int_{\fM_{g,\bsa}}\!\!\!\d \mu_{\mathfrak{b}\mathfrak{c},\beta\gamma}\, \langle V_{P_1}^{a_1} \cdots V_{P_n}^{a_n} \rangle \, \langle \hat{V}_{\hat{P}_1}^{a_1} \cdots \hat{V}_{\hat{P}_n}^{a_n} \rangle=\int_{\bfM_{g,\bsa}} \!\!\!\std_\pm(\fM_{g,\bsa}) \mathrm{e}^{c_1(\fL_{\bP,\bsa}^{(b)})} \begin{cases}
    (-1)^\zeta & \text{0A}\ , \\
    1 & \text{0B}\ .
\end{cases}  \label{eq:main identification}
\ee
In the type 0A case, we necessarily have only NS-punctures, i.e.\ $\bsa=\boldsymbol{1}$.

\section{Intersection theory and topological recursion} \label{sec:intersection theory}
Our main task will now be to evaluate the right hand side further. We will show that it can be reduced to intersection theory on the ordinary moduli space $\bM_{g,n}$. This will allow us to relate it to topological recursion.
\subsection{Reducing to the moduli space of spin curves}

\paragraph{Split supermanifolds.} Let us recall the notion of a split supermanifold. Let $M$ be a bosonic manifold and $F$ a vector bundle over it. Then one can reverse the parity of the fibers of $F$ and obtain a supermanifold $\Pi F$. On $\Pi F$, the bosonic coordinates $x^1,\dots,x^{n_\text{b}}$ parametrize $M$, while the directions along the fibers constitute the fermionic coordinates $\theta^1,\dots,\theta^{n_\text{f}}$. A supermanifold of this form is said to be split. This structure is not generic, since the coordinate transition functions act separately on the bosonic coordinates and are linear in the fermionic coordinates. For more general supermanifolds, the transition functions of the bosonic coordinates can involve any superfunction, such as fermion bilinears.

For real smooth supermanifolds, splitness is the generic situation: Batchelor's theorem \cite{Batchelor:1979tams} states that any such supermanifold $\fM$ with bosonic subspace $M$ is isomorphic to $\Pi F$ for some choice of vector bundle $F$ (although not canonically so). The analogous statement fails in the complex-analytic category \cite{Witten:2012bg, Witten:2012ga}. In particular, $\fM_{g,\bsa}$ is generically not split as a complex supermanifold \cite{Donagi:2013dua}. In superstring perturbation theory, this fact creates substantial additional complications. However, we will now show that for the purposes of computing the index, one can replace $\fM_{g,\bsa}$ by its split model $\fM_{g,\bsa}^\text{split}$.

\paragraph{Reduction to the split model.} The bosonic submanifold of $\fM_{g,\bsa}$ is $\cM_{g,\bsa}^\text{spin}$, embedded as the locus where all fermionic coordinates are set to zero. We can similarly truncate the coordinate transition functions to first order in $\theta^i$, which gives the split model $\fM_{g,\bsa}^\text{split}=\Pi \cF$, where $\cF$ is the bundle of fermionic tangent directions.\footnote{For bosonic moduli spaces and quantities on them, we use a mathcal font.} As already mentioned above, $\fM_{g,\bsa}$ is in general not isomorphic to $\fM_{g,\bsa}^\text{split}$. We now explain that index theory is blind to this fact, and for the purpose of computing the index we may replace $\fM_{g,\bsa}$ by $\fM_{g,\bsa}^\text{split}$.

Let $F(\boldsymbol{z} \mid \boldsymbol{\theta})$ be a holomorphic section of $\fL_{\bP,\bsa}^{(b)}$. It admits an expansion into the fermionic coordinates of the form
\be 
 F(\boldsymbol{z}\mid\boldsymbol{\theta}) = f_0(\boldsymbol{z}) + f_i(\boldsymbol{z})\theta^i + f_{ij}(\boldsymbol{z})\theta^i\theta^j + \cdots\ . \label{eq:F theta expansion}
\ee
On $\fM_{g,\bsa}$, the coordinate transition functions can depend on the fermionic coordinates, which mixes the different components of this expansion (preserving their parity). The individual coefficients $f_{i_1 \cdots i_k}$ are therefore not globally well-defined. What is globally well-defined is the statement that $F(\boldsymbol{z} \mid \boldsymbol{\theta})$ vanishes to order $k$ in the fermionic coordinates, meaning that $f_0, f_i, \ldots, f_{i_1 \cdots i_{k-1}}$ all vanish and the first non-vanishing coefficient is $f_{i_1 \cdots i_k}$. This is because transition functions can only increase the $\theta$-degree, never decrease it. It follows that there is a filtration
 \be
 \fL_{\bP,\bsa}^{(b)}\supset\mathcal{J}\fL_{\bP,\bsa}^{(b)}\supset\mathcal{J}^2\fL_{\bP,\bsa}^{(b)}\supset\cdots\supset\mathcal{J}^{n_\text{f}+1}\fL_{\bP,\bsa}^{(b)}=0\ ,
 \ee
where $\mathcal{J}$ is the ideal sheaf of holomorphic functions on $\fM_{g,\bsa}$ that vanish on $\cM_{g,\bsa}^\text{spin}$, $\mathcal{J}^k$ is its $k$-th power (functions vanishing to $k$-th order in the fermionic coordinates), and $\mathcal{J}^k \fL_{\bP,\bsa}^{(b)}$ is the subsheaf of holomorphic sections of $\fL_{\bP,\bsa}^{(b)}$ vanishing to $k$-th order.

 For the successive quotients, we have the isomorphism
 \be 
\mathcal{J}^{k}\fL_{\bP,\bsa}^{(b)}/\mathcal{J}^{k+1}\fL_{\bP,\bsa}^{(b)} \cong \mathcal{L}_{\bP,\bsa}^{(b)} \otimes \textstyle\bigwedge^{\! k} \cF^*\ ,
 \ee 
where $\smash{\mathcal{L}_{\bP,\bsa}^{(b)} \cong \fL_{\bP,\bsa}^{(b)}\big|_{\cM_{g,\bsa}^\text{spin}}}$ and $\bigwedge^{\! k}$ denotes the $k$-th exterior product. Indeed, on $\fM_{g,\bsa}^\text{split}$ the coefficients $f_{i_1 \cdots i_k}$ would be sections of $\smash{\mathcal{L}_{\bP,\bsa}^{(b)} \otimes \textstyle\bigwedge^{\! k} \cF^*}$. On $\fM_{g,\bsa}$ this is not true, but the quotient by $\mathcal{J}^{k+1}\fL_{\bP,\bsa}^{(b)}$ removes precisely the mixing of $f_{i_1 \cdots i_k}$ with the other coefficient functions, since those are set to zero. This means that there is a short exact sequence (of sheaves)
 \be
0\longrightarrow\mathcal{J}^{k+1}\fL_{\bP,\bsa}^{(b)}\longrightarrow\mathcal{J}^k\fL_{\bP,\bsa}^{(b)}\longrightarrow\mathcal{L}_{\bP,\bsa}^{(b)}\otimes{\textstyle\bigwedge^{\! k}}\cF^*\;\longrightarrow\;0\ ,\label{eq:filtration SES}
 \ee
The Euler characteristic is additive over exact sequences, i.e.
\be 
\chi\big(\fM_{g,\bsa},\mathcal{J}^k\fL_{\bP,\bsa}^{(b)}\big)-\chi\big(\fM_{g,\bsa},\mathcal{J}^{k+1}\fL_{\bP,\bsa}^{(b)}\big)=\chi\big(\cM_{g,\bsa}^\text{spin},\mathcal{L}_{\bP,\bsa}^{(b)}\otimes{\textstyle\bigwedge^{\!k}}\cF^*\big)\ .
\ee
If one takes the sum over $k$, the left hand side telescopes. Since $\mathcal{J}^{n_\text{f}+1}=0$, it follows that
 \be
\chi\big(\fM_{g,\bsa},\fL_{\bP,\bsa}^{(b)}\big)=\sum_{k=0}^{n_\text{f}}\chi\big(\cM_{g,\bsa}^\text{spin},\mathcal{L}_{\bP,\bsa}^{(b)}\otimes{\textstyle\bigwedge^{\!k}}\cF^*\big)=\chi\big(\fM_{g,\bsa}^\text{split},\fL_{\bP,\bsa}^{(b)}\big)\ .\label{eq:chi equality}
 \ee
Thus, for the purposes of computing the index, we may replace $\fM_{g,\bsa}$ by its split model $\fM_{g,\bsa}^\text{split}$.

\paragraph{Reduction of the super Todd class.} After passing to the split model, we can reduce the right hand side of \eqref{eq:main identification} to the bosonic subspace $\cM_{g,\bsa}^\text{spin}$. For the Chern class of the line bundle $\fL_{\bP,\bsa}^{(b)}$ this is simple. It represents the conformal anomaly of a supersymmetric conformal field theory on a spin surface. This anomaly is just the usual conformal anomaly when we restrict it to $\cM_{g,\bsa}^\text{spin}$, which leads to the line bundle $\mathcal{L}_{\bP,\bsa}^{(b)}$ that we already used above. Similarly, the Arf invariant is purely topological. The only non-trivial part is the reduction of the super Todd class.

To make sense of the super Todd class appearing in \eqref{eq:main identification}, we work with the middle expression of \eqref{eq:chi equality}. We can apply the ordinary index theorem on $\cM_{g,\bsa}^\text{spin}$ with the vector bundle $\mathcal{L}_{\bP,\bsa}^{(b)} \otimes \textstyle \bigwedge^{\! \bullet} \cF^*$, where $\bigwedge^{\! \bullet} \cF^* = \bigoplus_{k \ge 0} \bigwedge^{\! k} \cF^*$. It features the Chern character, which is defined as $\ch(E)=\sum_i \mathrm{e}^{c_1(L_i)}$ for $E=\bigoplus_i L_i$ a sum of line bundles.\footnote{By the splitting principle, it suffices to verify identities for characteristic classes on direct sums of line bundles.}

In other words, $\std_\pm (\fM_{g,\bsa})$ reduces to $\td(\cM_{g,\bsa}^\text{spin}) \, \ch_\pm( {\bigwedge}^\bullet \cF^*)$ on the spin moduli space, where the subscript $\pm$ corresponds to the optional insertion of $(-1)^\text{F}$. The bundle $\bigwedge^\bullet \cF^*$ is naturally $\ZZ_2$-graded as $\bigwedge^{\!\text{even}} \cF^* \oplus \bigwedge^{\!\text{odd}} \cF^*$, whose summands contain the bosonic and fermionic sections respectively. The passage to the split model explained above works separately on each summand, so the $\ZZ_2$-grading remains well-defined. We thus define
\be 
\ch_\pm(\textstyle\bigwedge^{\!\bullet}\cF^*) := \ch(\textstyle\bigwedge^{\!\text{even}} \cF^*) \pm \ch(\textstyle\bigwedge^{\!\text{odd}} \cF^*) \ ,
\ee
so that $\ch_+$ is the ordinary Chern character, while $\ch_-$ implements the inclusion of $(-1)^\text{F}$. 

To summarize, we can rewrite the right hand side of \eqref{eq:main identification} as
\be 
\int_{\bM_{g,\bsa}^\text{spin}} \td(\cM_{g,\bsa}^\text{spin})\, \ch_\pm(\textstyle\bigwedge^{\!\bullet}\cF^*) \, \mathrm{e}^{c_1(\mathcal{L}_{\bP,\bsa}^{(b)})}\begin{cases}
    (-1)^\zeta & \text{0A}\ , \\
    1 & \text{0B}\ .
\end{cases} \label{eq:index theorem spin moduli space}
\ee
This is also the form of the index theorem that appears in the mathematical literature \cite{Voronov:1990aa}.
\subsection{Reducing to standard intersection numbers}
The next step in our chain of logic is to compute the characteristic classes entering \eqref{eq:index theorem spin moduli space} more explicitly in terms of the standard $\kappa$ and $\psi$ classes. This can be done straightforwardly from standard results.

\paragraph{Pushforward vector bundles.} Both the bosonic and fermionic tangent directions $\cB=T \cM_{g,\bsa}^\text{spin}$ and $\cF$ are natural vector bundles over the spin moduli space. In fact, the cotangent space $\cB^*$ is spanned by quadratic differentials with possible simple poles at the puncture, i.e.\ $\cB_\Sigma^* \cong H^0(\Sigma,L_\text{b})$ with $L_\text{b} \cong \omega^2 (\sum_i [z_i]) \cong \omega_\text{log}^2(-\sum_i [z_i])$. In order to make contact with mathematical literature, we will express these bundles in the following through the logarithmic canonical bundle $\omega_\text{log} \cong \omega(-\sum_i [z_i])$. Notice that $H^1(\Sigma,L_\text{b})$ vanishes. The dimension of the tangent space can be computed by the Riemann-Roch theorem, which gives $3g-3+n$ as expected. Similarly, $\cF^*$ is spanned by $\frac{3}{2}$-differentials with possible simple poles at NS-punctures and inverse square root singularities at R-punctures. In other words, $\cF^* \cong H^0(\Sigma,L_\text{f})$ with $L_\text{f}^2 \cong \omega^3(\sum_i (1+a_i)[z_i]) \cong \omega_\text{log}^3(-\sum_i(2-a_i) [z_i])$. On $\cM_{g,\bsa}^\text{spin}$, the line bundle $L_\text{f}$ is by definition well-defined, since $\cM_{g,\bsa}^\text{spin}$ parametrizes punctured surfaces together with generalized spin structures. The degree of $L_\text{f}$ equals $\deg L_\text{f} = \frac{1}{2}(3(2g-2)+n+|\bsa|)$. Therefore, by Riemann-Roch, $H^0(\Sigma,L_\text{f})$ has dimension $2g-2+\frac{1}{2}(n+|\bsa|)$, which is indeed the expected fermionic dimension of super moduli space.

Both of these vector bundles have the form of a pushforward as follows. Consider the universal spin curve $\mathcal{C}_{g,\bsa}^\text{spin}$. This is the total space of the bundle, where the fiber at a point in $\cM_{g,\bsa}^\text{spin}$ is the spin surface itself parametrized by that point. Let 
\be 
\pi: \mathcal{C}_{g,\bsa}^\text{spin} \longrightarrow \cM_{g,\bsa}^\text{spin}
\ee 
be the projection that forgets the point on the surface in the fiber. 
We can think of $L_\text{b}$ and $L_\text{f}$ as living on the universal curve since it is a line bundle on the surface itself over the whole family of spin surfaces parametrized by $\cM_{g,\bsa}^\text{spin}$.
Then by definition
\be 
\cB^* \cong H^0(\Sigma,L_{\text{b}})=\pi_* L_{\text{b}}\ , \qquad \cF^* \cong H^0(\Sigma,L_{\text{f}})=\pi_* L_{\text{f}}
\ee
is the pushforward of the line bundles $L_\text{b}$ and $L_\text{f}$ to moduli space.\footnote{In general, one has to correct by higher cohomology groups and obtains a derived sheaf, but these vanish in the present case.}
In such situations, one can compute characteristic classes of $\pi_* L_{\text{b/f}}$ by the Grothendieck-Riemann-Roch theorem, applied to the map $\pi$. This computation was carried out by Chiodo in great generality \cite{Chiodo:2008}.

\paragraph{Chiodo formula.} Chiodo's formula precisely computes the Chern character of pushforwards of vector bundles. This can be done on moduli spaces of $r$-spin structures. In our case, ordinary spin structures correspond to $r=2$. Consider a line bundle $L$ such that $L^2 \cong \omega_\text{log}^s (-\sum_{i=1}^n m_i [z_i])$. Then 
\begin{multline} 
\ch_d(\pi_* L)=\frac{1}{(d+1)!} \bigg(B_{d+1}(\tfrac{s}{2})\kappa_d-\sum_{i=1}^n B_{d+1}(\tfrac{m_i}{2}) \psi_i^d \\ 
+\sum_{q=0,1} B_{d+1}(\tfrac{q}{2}) (\xi_q)_*\Big(\frac{\psi_\bullet^{d}-(-\psi_\circ)^d}{\psi_\bullet+\psi_\circ}\Big)\bigg)\ . \label{eq:Chiodo formula}
\end{multline}
Let us explain the notation. $B_n(x)$ denotes Bernoulli polynomials. $\kappa_d$ and $\psi_i$ are the basic cohomology classes on the moduli space of spin curves. The definition is analogous to their corresponding definition on the moduli space of curves. Lastly, $\xi$ is the inclusion of the union of all boundary divisors in $\bM_{g,\bsa}^\text{spin}$. The boundary divisors furthermore fall into two distinct sets depending on whether the nodal points are NS-punctures ($q=1$) or R-punctures ($q=0$). This leads to the two embedding maps $\xi_q$. $(\xi_q)_*$ includes the relevant automorphism factors of boundary divisors.
Finally, $\psi_\bullet$ and $\psi_\circ$ are the two $\psi$-classes associated to the two nodes. 
Notice in particular that as a consequence of $\kappa_0=2g-2+n$
\be 
\ch_0(\pi_* L)=\frac{1}{2}(s-1)(2g-2+n)-\frac{1}{2}\sum_{i=1}^n (m_i-1)
\ee
recovers the expected dimension of the bundle as computed by Riemann-Roch.

Let us apply \eqref{eq:Chiodo formula} to the two cases of interest. It will be sufficient for most cases to only consider $\ch_1(\pi_* L_\text{b/f})$ and $\ch_{2k}(\pi_* L_\text{b/f})$. For these cases dramatic simplifications occur. For $L_\text{b}$, we use $s=4$ and $m_i=2$. For $L_\text{f}$, we use $s=3$ and $m_i=2-a_i$. We obtain
\begin{subequations}
\begin{align} 
\ch_1(\cB^*)&=\frac{13}{12}\kappa_1-\frac{1}{12}\sum_i \psi_i+\frac{1}{12} \delta_0-\frac{1}{24} \delta_1\ , \label{eq:ch1 B}\\
\ch_1(\cF^*)&=\frac{11}{24}\kappa_1+\sum_i\Big(\frac{1}{24}-\frac{1}{8} \delta_{a_i,0}\Big) \psi_i+\frac{1}{12} \delta_0-\frac{1}{24} \delta_1\ , \label{eq:ch1 F}\\
\ch_{2k}(\cB^*)&=\frac{\kappa_{2k}}{(2k)!}\ , \\
\ch_{2k}(\cF^*)&=\frac{\kappa_{2k}}{2^{2k}(2k)!}\ .
\end{align} \label{eq:Chern character B and F}%
\end{subequations}
Here, we wrote $\delta_q=(\xi_q)_*(1)$ for the boundary class. 
Remarkably, all $\psi$-classes and boundary contributions to higher even Chern characters vanish because the corresponding Bernoulli polynomials vanish. We also used the identities $B_{2k+1}(2)=2k+1$ and $B_{2k+1}(\frac{3}{2})=(2k+1)2^{-2k}$, which follow from the generating function $\frac{t\, \mathrm{e}^{xt}}{\mathrm{e}^t-1}=\sum_{n \ge 0} B_n(x) \frac{t^n}{n!}$. For the dual bundle, we can use that $\ch_d(E^*)=(-1)^d \ch_d(E)$ for any vector bundle $E$.
\paragraph{Symmetric function exercises.} 
From the Chern character of $\cB$ and $\cF$, we can compute the Todd class and the character $\ch_\pm(\bigwedge^\bullet \cF^*)$. Let us start with the former. For this one can use the splitting principle. Assume momentarily that $\cB=\bigoplus_i \mathcal{L}_i$. Let $x_i=c_1(\mathcal{L}_i)$. Then $\ch_d(\cB)=\frac{1}{d!} \sum_i x_i^d$. The Todd class is defined as
\be 
\td(\cB)=\prod_i \frac{x_i}{1-\mathrm{e}^{-x_i}}=\exp\bigg(-\sum_{d \ge 1} \sum_{i} \frac{B_d x_i^d}{d\, d!}\bigg)=\exp\bigg(-\sum_{d \ge 1} \frac{B_d}{d} \, \ch_d(\cB) \bigg)\ . \label{eq:td B}
\ee
Except for $B_1=-\frac{1}{2}$, odd Bernoulli numbers vanish and thus this class indeed only contains the Chern characters computed above. We obtain
\be 
\td(\cB)=\exp\bigg(-\frac{13}{24}\kappa_1+\frac{1}{24}\sum_i \psi_i-\frac{1}{24} \delta_0+\frac{1}{48} \delta_1-\sum_{k \ge 1} \frac{B_{2k}\kappa_{2k}}{2k (2k)!}\bigg)\ .
\ee
Let us turn to the fermionic class. We begin with $\ch_+(\bigwedge^\bullet\cF^*)$. Assume again momentarily that $\cF=\bigoplus_i \mathcal{L}_i$ and $x_i=c_1(\mathcal{L}_i)$. Then
\begin{align} 
\ch_+({\bigwedge}^\bullet \cF^*)&=\prod_i (1+\mathrm{e}^{-x_i}) \\
&=2^{\rk \cF} \exp\bigg(\sum_{d\ge 1}\sum_i \frac{B_d(2^d-1)x_i^d}{d\, d!}\bigg) \\
&=2^{\rk \cF}\exp\bigg(\sum_{d\ge 1} \frac{B_d(2^d-1)}{d}\, \ch_d(\cF)\bigg)\\
&=2^{\rk \cF} \exp\bigg(\frac{11}{48} \kappa_1+\sum_i \Big(\frac{1}{48}-\frac{1}{16} \delta_{a_i,0}\Big) \psi_i+\frac{1}{24} \delta_0-\frac{1}{48} \delta_1\nonumber\\
&\qquad\qquad\qquad\qquad\qquad\qquad\qquad\qquad+\sum_{k \ge 1} \frac{B_{2k}(1-2^{-2k}) \kappa_{2k}}{2k(2k)!}\bigg)\ . \label{eq:ch+ F}
\end{align}
Remarkably, once we combine \eqref{eq:td B} and \eqref{eq:ch+ F}, all boundary classes cancel out. There are also additional cancellations in the $\psi$-classes as well as in the kappa classes. We obtain
\be 
\td(\cB)\ch_+({\bigwedge}^\bullet \cF^*)=2^{\rk \cF} \exp\bigg(-\frac{15}{48}\kappa_1+\frac{1}{16}\sum_i \delta_{a_i,1}\psi_i-\sum_{k \ge 1} \frac{B_{2k} \kappa_{2k}}{2k(2k)!2^{2k}}\bigg)\ . \label{eq:tdB ch+ F combination}
\ee
Let us turn to $\ch_-(\bigwedge^\bullet \cF^*)=\prod_i (1-\mathrm{e}^{-x_i})$. This behaves differently because each term in the product does not have a constant piece. Therefore it is a cohomology class whose lowest degree piece is in dimension $\rk \cF$ and does not take the form of an exponential. We will instead factor out the piece $\prod_i x_i$, which equals the top Chern class or Euler class $e(\cF)$ of $\cF$. We will see below that one can deal with it using known results. Therefore, we compute
\begin{align}
    \ch_-({\bigwedge}^\bullet \cF^*)&=e(\cF) \exp\bigg(\sum_{d \ge 1} \sum_i \frac{B_d x_i^d}{d\, d!} \bigg)\\
    &=e(\cF)\exp\bigg(\sum_{d \ge 1} \frac{B_d}{d} \ch_d(\cF)\bigg)\\
    &=e(\cF)\exp\bigg(\frac{11}{48} \kappa_1+\sum_i \bigg(\frac{1}{48}-\frac{1}{16} \delta_{a_i,0}\bigg) \psi_i+\frac{1}{24} \delta_0-\frac{1}{48} \delta_1\nonumber\\
    &\qquad\qquad\qquad\qquad\qquad\qquad\qquad\qquad\qquad+\sum_{k \ge 1} \frac{B_{2k}\kappa_{2k}}{2k(2k)!2^{2k}} \bigg)\ . \label{eq:ch- F}
\end{align}
When combining the bosonic and the fermionic piece, we observe a similar cancellation as in the case above. We obtain
\be 
\td(\cB)\ch_-({\bigwedge}^\bullet \cF^*)=e(\cF) \exp\bigg(-\frac{15}{48}\kappa_1+\frac{1}{16}\sum_i \delta_{a_i,1}\psi_i-\sum_{k \ge 1} \frac{B_{2k} (1-2^{-2k})\kappa_{2k}}{2k(2k)!}\bigg)\ . \label{eq:tdB ch- F combination}
\ee
The cancellation of boundary classes is consistent with the expectation that the index theorem \eqref{eq:naive index theorem} receives no boundary corrections in the supersymmetric case. Let us remark that this cancellation does not occur in the bosonic case. In \cite{Collier:2023cyw} boundary classes were discarded but this procedure was not well-justified. We plan to return to this issue in more generality in the future \cite{VMS_universality}.
\paragraph{The line bundle $\mathcal{L}_{\bP,\bsa}^{(b)}$.} The final ingredient needed to evaluate the RHS is an explicit expression for $c_1(\mathcal{L}_{\bP,\bsa}^{(b)})$. $c_1(\mathcal{L}_{\bP,\bsa}^{(b)})$ captures the conformal anomaly. As already mentioned above, the conformal anomaly of a superconformal field theory is just the conformal anomaly of the underlying CFT. Therefore, the formula is unchanged from the bosonic case \cite{Belavin:1986cy, Collier:2023cyw}. It reads
\be 
c_1(\mathcal{L}_{\bP,\bsa}^{(b)})=\frac{c}{24} \kappa_1+\sum_i \Big(h_i-\frac{c}{24}\Big) \psi_i\ . \label{eq:c_1 L formula}
\ee 
The terms proportional to $c$ capture the conformal anomaly, while $h_i \psi_i$ carries the weight of the $i$-th vertex operator under conformal transformations. 

We see that the coefficients of the $\psi$ and $\kappa$ classes appearing in \eqref{eq:ch1 B} and \eqref{eq:ch1 F} are no accident. They precisely account for the contribution of the ghosts to the central charge and the ground state energy, respectively.

With the standard parametrizations \eqref{eq:c h parametrization}, we end up with the following formulas,
\begin{multline}
\td(\cB)\ch_\sigma ({\bigwedge}^\bullet \cF^*)\mathrm{e}^{c_1(\mathcal{L}_{\bP,\bsa}^{(b)})}= \exp\bigg(\frac{b^2+b^{-2}}{8}\kappa_1+\frac{1}{2}\sum_i P_i^2\psi_i\bigg) \\
\times \begin{cases}
    2^{\rk \cF}\exp\Big(-\sum_{k \ge 1} \frac{B_{2k}\kappa_{2k}}{2k(2k)!2^{2k}}\Big)\ , &\sigma=+\ , \\
    e(\cF)\exp\Big(-\sum_{k \ge 1} \frac{B_{2k} (1-2^{-2k})\kappa_{2k}}{2k(2k)!}\Big)\ , &\sigma=-\ .
\end{cases} \label{eq:summary integrand moduli space of spin curves}
\end{multline}
Remarkably, this expression barely distinguishes between NS-sector and R-sector vertex operators. They still enter via $2^{\rk \cF}$ and $e(\cF)$, but are otherwise indistinguishable. 
\paragraph{Pushforward to ordinary moduli space.} So far, \eqref{eq:summary integrand moduli space of spin curves} is still a class on the moduli space of spin curves. Consider the map
\be 
p: \bM_{g,\bsa}^\text{spin} \longrightarrow \bM_{g,n}
\ee
that forgets the spin structure. Summing over spin structures means that we are computing the pushforward of the cohomology class \eqref{eq:summary integrand moduli space of spin curves}, possibly with the insertion of the Arf invariant $(-1)^\zeta$. This is actually quite simple. $p$ is a $2^{2g}$-fold covering. It is unramified over the interior of moduli space. Since \eqref{eq:summary integrand moduli space of spin curves} does not involve boundary classes (except in $e(\cF)$), one can essentially just multiply the expression with the degree. For the case of 0B$^+$, this leads to the additional factor of $2^{2g-1}$. The extra factor of $\frac{1}{2}$ originates from the orbifold structure of $\bM_{g,\bsa}^\text{spin}$. Each spin surface carries an extra $\ZZ_2$ automorphism group acting on the line bundle. This factor of $\frac{1}{2}$ may also be viewed more physically as originating as the result of gauging the $\ZZ_2$ that leads to the sum over spin structures. For type 0A$^+$, the factor is instead $\frac{1}{2}(\frac{1}{2}(2^{2g}+2^g)-\frac{1}{2}(2^{2g}-2^g))=2^{g-1}$, where we use that there are $\frac{1}{2}(2^{2g}+2^g)$ even and $\frac{1}{2}(2^{2g}-2^g)$ odd spin structures. 

For the 0A/B$^-$ theories, we have to account for the Euler class $e(\cF)$. We can use the projection formula. Since the remaining part of the integrand just uses $\kappa$ and $\psi$-classes, it has the form of a pullback $p^*\alpha$, where $\alpha$ is the same expression of $\kappa$ and $\psi$-classes, but on the ordinary moduli space of curves. Then the projection formula gives
\be 
p_*(e(\cF) p^* \alpha)=(p_* e(\cF))\alpha\ .
\ee
Following Norbury \cite{Norbury:2017eih}, we define the Theta class as $\Tgn=(-1)^n 2^{g-1+n}\, p_* e(\cF)$. The normalization ensures simple factorization properties. It was first defined in \cite{Norbury:2017eih} in the case of $\bsa=\boldsymbol{1}$ and is part of a larger family of generalized Theta classes known as Chiodo classes. The generalizations for other choices of $\bsa$ were considered in \cite{Chidambaram:2022cqc, Alexandrov:2024kuj}.

For type 0A$^-$, we need to compute the pushforward signed by the Arf invariant. In this case, we necessarily have $\bsa=\boldsymbol{1}$. We define correspondingly the Arf theta class as
\be
\Tgn^\zeta=(-1)^n 2^{g-1+n}\big(p^+_* e(\cF)-p^-_* e(\cF)\big)\ ,
\ee
where $p^\pm:\bM_{g,\boldsymbol{1}}^{\text{spin},\pm} \longrightarrow \bM_{g,n}$ is the forgetful map from the component of even and odd Arf invariant. It is a special case of a spin Chiodo class as defined in \cite{Giacchetto:2021qss}.

\subsection{The type \texorpdfstring{0A/B$^+$}{0A/B+} theories}
Let us begin by discussing the theories without $(-1)^\text{F}$ insertions. Their treatment is much simpler.
\paragraph{Duality between bosonic VMS, 0A$^+$ and 0B$^+$ super VMS.} We derived the following intersection number expression on ordinary moduli space,
\begin{align}
\mathsf{V}_{g,\bsa}^{(b),\text{0A/B}^+}(\bP)&=\begin{cases} 2^{3g-3+n}\ , &\text{0A}^+\\
2^{4g-3+\frac{1}{2} (|\bsa|+n)}\ , &\text{0B}^+
\end{cases} \nonumber\\
&\qquad\times \int_{\bM_{g,n}}\exp\bigg(\frac{b^2+b^{-2}}{8}\kappa_1+\frac{1}{2}\sum_i P_i^2\psi_i-\sum_{k \ge 1} \frac{B_{2k}\kappa_{2k}}{2k(2k)!2^{2k}}\bigg) \\
&=\mathsf{V}_{g,n}^{(b)}(\bP)\begin{cases} 1\ , &\text{0A}^+\\
2^{g+\frac{1}{2} (|\bsa|-n)}\ , &\text{0B}^+
\end{cases}\ .
\end{align}
In passing to the second line, we used that the integrand is related to the bosonic one by uniform rescaling of the classes according to their weight, where
\be 
\mathsf{V}_{g,n}(\bP)=\int_{\bM_{g,n}}\exp\bigg(\frac{b^2+b^{-2}}{4}\kappa_1+\sum_i P_i^2\psi_i-\sum_{k \ge 1} \frac{B_{2k}\kappa_{2k}}{2k(2k)!}\bigg)\ .
\ee
Since the integral picks out the top form, this leads to an additional relative factor of $2^{-\dim_\CC \bM_{g,n}}$.

Thus we conclude that type 0A$^+$ super VMS is perturbatively equivalent to bosonic VMS. For type 0B$^+$, this is also true with small modifications. We see that R-sector vertex operators produce equivalent string amplitudes to NS-sector vertex operators up to a factor of $\sqrt{2}$, i.e.\ the worldsheet vertex operator $\sqrt{2} \mathcal{O}^\text{R}_P$ has the same amplitudes as $\mathcal{O}^\text{NS}_P$. This equivalence of vertex operators is very surprising from a worldsheet point of view. The model has worldsheet supersymmetry, but no notion of target space supersymmetry from a string theory point of view. Thus these vertex operators do not sit in the same multiplet and there is a priori not any obvious relation between the two. 
Moreover, the overall ratio $\mathsf{V}^{\text{0B}^+}_{g,n}/\mathsf{V}^{\text{0A}^+}_{g,n}=2^g$ reflects the ratio of the total number of spin structures $2^{2g}$ to the signed count $\sum_\text{spin}(-1)^\zeta=2^g$ on a genus-$g$ surface \cite{Stanford:2019vob}.
Up to these normalization factors, all three theories are therefore perturbatively equivalent.
\paragraph{The dual matrix integral.} In \cite{Collier:2023cyw}, a duality between the VMS and a dual matrix integral was established. We will restrict ourselves to the formulation of the loop equations of the matrix integral in terms of topological recursion. The topological recursion starts from the basic data
\be 
\omega_{0,1}(z)=4 \sqrt{2} \pi \sin(2\pi b z) \sin(2\pi b^{-1} z) \d z\ , \qquad \omega_{0,2}(z_1,z_2)=\frac{\d z_1\, \d z_2}{(z_1-z_2)^2}\ . \label{eq:bosonic spectral curve}
\ee
Higher differentials can be computed from the recursion
\begin{multline} 
\omega_{g,n}(z_1,\boldsymbol{z})=\Res_{z=0} K(z_1,z) \bigg(\omega_{g-1,n+1}(z,-z,\boldsymbol{z}) \\ +\sum_{h=0}^g \sideset{}{'}\sum_{I \subset \boldsymbol{z}} \omega_{h,|I|+1}(z,I) \omega_{g-h,|I^c|+1} (-z,I^c)\bigg)\ . \label{eq:topological recursion}
\end{multline}
Here $\boldsymbol{z}=\{z_2,\dots,z_n\}$. The prime on the sum indicates that we leave out all terms that involve $\omega_{0,1}$, i.e.\ $(h,I)=(0,\emptyset)$ and $(h,I)=(g,\boldsymbol{z})$.
The recursion kernel is given by
\be 
K(z_1,z)=-\frac{1}{4 \omega_{0,1}(z)}\int_{-z}^z \omega_{0,2}(z_1,-)=-\frac{z}{8\sqrt{2}\pi (z_1^2-z^2)\sin(2\pi b z)\sin(2\pi b^{-1}z)}\ . \label{eq:recursion kernel}
\ee
The relation to the worldsheet observables is given by
\be 
\omega_{g,n}(\boldsymbol{z})=\int_0^\infty \prod_j (-4 \sqrt{2}\pi P_j \, \d P_j\, \mathrm{e}^{-4\pi z_j P_j}) \mathsf{V}_{g,n}^{(b)}(\bP) \prod_{j=1}^n \d z_j\ .
\ee
\paragraph{The disk partition function.} $\omega_{0,1}(z)$ is related to the density of eigenvalues in the matrix integral. It was observed in \cite{Collier:2023cyw} that it equals the Cardy density of states of a 2d CFT of central charge $c$. The Cardy density of states is the modular S-kernel of the vacuum Virasoro character. For the Virasoro algebra, $\chi_P(\tau)=q^{P^2} \eta(\tau)^{-1}$ and $\chi_\text{vac}(\tau)=q^{-\frac{c-1}{24}}(1-q) \eta(\tau)^{-1}=\chi_{\frac{i}{2}(b+b^{-1})}(\tau)-\chi_{\frac{i}{2}(b-b^{-1})}(\tau)$. We then have
\be 
\chi_\text{vac}(-\tfrac{1}{\tau})=\int_0^\infty \d P \, 2\sqrt{2} \sinh(2\pi b P) \sinh(2\pi b^{-1} P) \chi_P(\tau)\ .
\ee
For the supersymmetric case, there are three different non-vanishing vacuum characters. We have
\begin{subequations} 
\begin{align}
\chi_P^{\text{NS}^+}(\tau)&=q^{\frac{1}{2}P^2}\vartheta_3(\tau)^{\frac{1}{2}}\eta(\tau)^{-\frac{3}{2}}\ , & \chi_\text{vac}^{\text{NS}^+}(\tau)&=\chi_{\frac{i}{2}(b+b^{-1})}^{\text{NS}^+}(\tau)-\chi_{\frac{i}{2}(b-b^{-1})}^{\text{NS}^+}(\tau)\ , \\
\chi_P^{\text{NS}^-}(\tau)&=q^{\frac{1}{2}P^2}\vartheta_4(\tau)^{\frac{1}{2}}\eta(\tau)^{-\frac{3}{2}}\ , & \chi_\text{vac}^{\text{NS}^-}(\tau)&=\chi_{\frac{i}{2}(b+b^{-1})}^{\text{NS}^-}(\tau)+\chi_{\frac{i}{2}(b-b^{-1})}^{\text{NS}^-}(\tau)\ , \\
\chi_P^{\text{R}^+}(\tau)&=2^{\frac{1}{2}} q^{\frac{1}{2}P^2}\vartheta_2(\tau)^{\frac{1}{2}}\eta(\tau)^{-\frac{3}{2}}\ , & \chi_\text{vac}^{\text{R}^+}(\tau)&=\frac{1}{2}\chi_{0}^{\text{R}^+}(\tau)\ .
\end{align}
\end{subequations}
The Ramond sector vacuum character does not contribute to the Cardy density. For the other two characters, we obtain
\begin{subequations}
\begin{align}
    \chi_\text{vac}^{\text{NS}^+}(-\tfrac{1}{\tau}) &= \int_0^\infty  \d P\, 4 \sinh(\pi b P) \sinh(\pi b^{-1} P) \chi_P^{\text{NS}^+}(\tau)\ , \label{eq:Neveu Schwarz Cardy density of states}\\
    \chi_\text{vac}^{\text{NS}^-}(-\tfrac{1}{\tau}) &= \int_0^\infty  \d P\, 2 \sqrt{2} \cosh(\pi b P) \cosh(\pi b^{-1} P) \chi_P^{\text{R}^+}(\tau)\ . \label{eq:Ramond Cardy density of states}
\end{align}
\end{subequations}
Thus, we have two universal Cardy densities of states in superconformal field theories. The type 0A/B$^+$ spectral curve is that of the bosonic VMS \eqref{eq:VMS density of states}, which is related to the NS-sector Cardy density \eqref{eq:Neveu Schwarz Cardy density of states} by a rescaling $z \to 2z$ and an overall factor of $\sqrt{2}$. The Ramond-sector density of states will be relevant for the 0A/B$^-$ theories.
\subsection{Type \texorpdfstring{0A$^-$}{0A-}}
Let us next discuss the type 0A$^-$ theory, which features the Arf theta class $\Tgn^\zeta$. It is an example of a spin Chiodo class as defined in \cite{Giacchetto:2021qss}. It can be computed explicitly in terms of standard classes on $\bM_{g,n}$.

Already at genus 0, all amplitudes involving $\Theta_{0,n}^\zeta$ vanish for dimensional reasons when $n \ge 3$: the Arf theta class has cohomological degree $n-2$, which exceeds $\dim_\CC \bM_{0,n}=n-3$, leaving no room for $\psi$-class insertions. The higher-genus vanishing is more subtle. We demonstrate in appendix~\ref{app:arf theta class} that any intersection number of the form
\be
\int_{\bM_{g,n}} \Tgn^\zeta \, \prod_{i=1}^n \psi_i^{d_i} \prod_{d \ge 1} \kappa_d^{m_d}
\ee
\emph{vanishes identically} for all $g \ge 1$. From the point of view of intersection theory, this cancellation is non-trivial.

It follows that the only non-vanishing perturbative amplitude of the type 0A$^-$ super VMS is the sphere 2-point function. In the language of string field theory, the sphere 2-point function defines the propagator (the kinetic term of the action), while higher-point sphere amplitudes define interaction vertices and higher-genus amplitudes give loop corrections. Since all of these vanish, the perturbative string field theory action is purely quadratic, i.e.\ the type 0A$^-$ theory is a free string field theory. It would be interesting to understand whether this theory receives non-perturbative corrections. From the matrix model perspective, vanishing of all non-trivial perturbative amplitudes indicates that this model may be dual to a matrix model without spectral edge, but non-perturbative data is required to determine which one.
\subsection{Type \texorpdfstring{0B$^-$}{0B-} with NS-sector vertex operators}
We now move on to the next theory, which will differ from bosonic VMS. We will start with type 0B$^-$ with only NS-sector vertex operators.
\paragraph{Theta class.} Let us recall that $\Tgn=(-1)^n 2^{g-1+n} p_* e(\cF)$. Thus, we have
\begin{multline} 
\mathsf{V}^{(b),\text{0B}^-}_{g,n}(\bP)=(-1)^n 2^{1-g-n}\int_{\bM_{g,n}} \Tgn \exp\bigg(\frac{b^2+b^{-2}}{8}\kappa_1+\frac{1}{2}\sum_i P_i^2\psi_i\\
-\sum_{k \ge 1} \frac{B_{2k} (1-2^{-2k})\kappa_{2k}}{2k(2k)!}\bigg)\ . \label{eq:intersection number 0B- NS sector vertex operators}
\end{multline}
This behaves very differently from bosonic VMS. $\Tgn$ is a class of pure degree $2g-2+n$, so the super quantum volume is a polynomial of degree $g-1$ in $P_i^2$ and $c$ (as opposed to $3g-3+n$ in the bosonic case). In particular, it vanishes identically at genus 0, where the fermionic dimension of moduli space is bigger than the bosonic dimension.
\paragraph{Topological recursion.} Intersection numbers of the form \eqref{eq:intersection number 0B- NS sector vertex operators} are also captured by topological recursion. They are governed by the so-called Br\'ezin-Gross-Witten (BGW) $\tau$-function \cite{Brezin:1980rk, Gross:1980he}, rather than the Kontsevich-Witten $\tau$-function that governs the intersection numbers of $\psi$-classes \cite{Witten:1990hr}. This was conjectured by Norbury \cite{Norbury:2017eih} and proven in \cite{Chidambaram:2022cqc, Norbury:2020vyi}, building on the topological recursion for the Bessel curve \cite{Do:2016odu, Chekhov:2017aot}. The dictionary is reviewed in appendix~\ref{app:hard topological recursion}. The coefficients of the spectral curve are determined through the equality \eqref{eq:power series identification}. In our case, $s_1=\frac{b^2+b^{-2}}{8}$ and $s_{2k}=-\frac{B_{2k}(1-2^{-2k})}{2k(2k)!}$, so 
\begin{align} 
\sum_{k=0}^\infty y_k(2k-1)!! u^k&=\exp\bigg(-\frac{b^2+b^{-2}}{8} u+\sum_{k\ge 1} \frac{B_{2k}(1-2^{-2k}) u^{2k}}{2k(2k)!}\bigg)\\
&=\mathrm{e}^{-\frac{b^2+b^{-2}}{8} u} \cosh\Big(\frac{u}{4}\Big)\ .
\end{align}
Extracting coefficients gives 
\be 
y_k=\frac{(-1)^k}{2(2k)!}\bigg[\Big(\frac{b+b^{-1}}{2}\Big)^{2k}+\Big(\frac{b-b^{-1}}{2}\Big)^{2k}\bigg]\ ,
\ee
and therefore
\be 
\omega_{0,1}(z)=\sum_{k=0}^\infty y_k z^{2k} \d z=\cos\big(\tfrac{1}{2} b z\big) \cos\big(\tfrac{1}{2}b^{-1}z\big)\, \d z\ .
\ee
We can rescale the coordinates by a factor of $4\pi$ since topological recursion is coordinate independent. Moreover, rescaling $\omega_{0,1}$ by a factor $\lambda$ will rescale $\omega_{g,n}$ in topological recursion by $\lambda^{2-2g-n}$ \cite{Eynard:2007kz}. Thus we can absorb the prefactor in \eqref{eq:intersection number 0B- NS sector vertex operators} by multiplying $\omega_{0,1}$ by another factor $\sqrt{2}$. The remaining prefactor $(-1)^n 2^{-n/2}$ in \eqref{eq:intersection number 0B- NS sector vertex operators} can be absorbed into vertex operator normalizations. We obtain
\be 
\omega_{0,1}(z)=4 \sqrt{2} \pi \cos(2\pi b z) \cos(2\pi b^{-1}z)\, \d z\ .
\ee
This is indeed the natural analogue of \eqref{eq:bosonic spectral curve}. It is related to the Ramond Cardy density of states \eqref{eq:Ramond Cardy density of states} in the same way as the type 0A/B$^+$ theories are related to the Neveu-Schwarz density of states.
\subsection{Type \texorpdfstring{0B$^-$}{0B-} with mixed vertex operators}
Finally, it remains to discuss the case of type 0B$^{-}$ with both NS- and R-sector vertex operators. 

\paragraph{Worldsheet normalizations.} This case presents us with an interesting puzzle. Let us start by considering the tree-level 3-point function with two R-sector and one NS-sector vertex operator inserted. The corresponding supermoduli space is $0 \mid 0$-dimensional. Thus the index theorem gives the answer 
\be 
\mathsf{V}_{0,3}^{(b),(1,0,0),\text{0B}^-}(P_1,P_2,P_3)=1\ ,
\ee
reflecting that there is only one state. This is the same as for the type 0B$^+$ theory.

From the worldsheet point of view, the answer was computed in \cite{Muhlmann:2025ngz, Rangamani:2025wfa} and was found to be \emph{vanishing}. This creates an apparent puzzle. We now argue that this is because of a normalization choice employed on the worldsheet. We follow the conventions of \cite{Muhlmann:2025ngz}. The GSO projected and physical Ramond-sector vertex operator in the $(-\frac{1}{2},-\frac{1}{2})$ picture takes the form
\be 
\mathcal{V}_p^\text{R}=\mathsf{R}(p) \mathfrak{c} \tilde{\mathfrak{c}} \mathrm{e}^{-\frac{1}{2}(\varphi+\tilde{\varphi})} \big(V_p^{\text{R},+} \widehat{V}_{ip}^{\text{R},-}+i V_p^{\text{R},-} \widehat{V}_{ip}^{\text{R},+}\big)\ .
\ee
Here, $p=i P$.
The spacelike super Liouville Ramond-sector vertex operators are denoted by $V_p^{\text{R},\pm}$ and the timelike vertex operators by $\widehat{V}_p^{\text{R},\pm}$.
The superscripts $\pm$ denote the two degenerate Ramond sector ground states of the two worldsheet theories. We can compute the string theory two-point function of these vertex operators. For this, one has to work out also the vertex operator in the picture number $(-\frac{3}{2},-\frac{3}{2})$, which we can take to be of the form (up to BRST exact terms)
\be 
\mathsf{R}(p) \mathfrak{c} \partial \mathfrak{c}\tilde{\mathfrak{c}}\bar{\partial} \tilde{\mathfrak{c}} \mathrm{e}^{-\frac{3}{2}(\varphi+\tilde{\varphi})} \big(V_p^{\text{R},+} \widehat{V}_{ip}^{\text{R},-}+i V_p^{\text{R},-} \widehat{V}_{ip}^{\text{R},+}\big)\ .
\ee
The string theory two-point function then reduces to the CFT two-point function, which gives
\begin{multline}
    \langle \big(V_{p_1}^{\text{R},+} \widehat{V}_{ip_1}^{\text{R},-}+i V_{p_1}^{\text{R},-} \widehat{V}_{ip_1}^{\text{R},+}\big)\big(V_{p_2}^{\text{R},+} \widehat{V}_{ip_2}^{\text{R},-}+i V_{p_2}^{\text{R},-} \widehat{V}_{ip_2}^{\text{R},+}\big) \rangle_{\text{sL}+\text{tL}} \\
    =\langle V_{p_1}^{\text{R},+} V_{p_2}^{\text{R},+} \rangle\langle \widehat{V}_{i p_1}^{\text{R},-} \widehat{V}_{ip_2}^{\text{R},-} \rangle-\langle V_{p_1}^{\text{R},-} V_{p_2}^{\text{R},-} \rangle\langle \widehat{V}_{i p_1}^{\text{R},+} \widehat{V}_{ip_2}^{\text{R},+} \rangle=0\ ,
\end{multline}
because the two-point functions between the two terms cancel out. 
Since the string theory two-point function vanishes, it is therefore not surprising that also higher-point functions are observed to vanish identically. From intersection theory, we obtain non-zero answers because the intersection-theoretic computation does not include the worldsheet normalization factor $\mathsf{R}(p)$, which is proportional to the inverse of the vanishing R-sector two-point function coefficient.

\paragraph{Ramond gas.} We can ask then what the meaning of the intersection theory expressions with Ramond-sector insertions is, especially in view of the dual matrix integral. A hint is provided by the conjecture of \cite{Stanford:2019vob} that adding a \emph{gas} of Ramond-punctures corresponds to a simple deformation of the spectral curve. Importantly, this gives an interpretation only to the Ramond vertex operators at $P=0$. We can think of this deformation as turning on R-R flux in the worldsheet theory. We do not know of a natural incorporation of the other Ramond-sector vertex operators.

This intuition has been made more precise in a sequence of recent mathematical works 
\cite{Chidambaram:2022cqc, Norbury:2023swp, Alexandrov:2024kuj}. Let us denote by $\Tgn^{\bsa}$ the generalization of the Theta class with $\bsa=(a_1,\dots,a_n)$ and $a_i \in \{0,1\}$ labelling the puncture type. The standard Theta class is $\Tgn^{(1^n)}$. It is shown in \cite{Chidambaram:2022cqc, Alexandrov:2024kuj} that the intersection numbers
\be 
\omega_{g,n}(z_1,\dots,z_n)=\sum_{m=0}^\infty \frac{t^m}{m!} \int_{\bM_{g,n+m}} \Theta_{g,n+m}^{(1^n,0^m)} \prod_{i=1}^n \d \Big(\frac{1}{1-\frac{\psi_i}{t-z_i} \partial_{z_i}} \frac{1}{z_i-t} \Big)
\ee
are also produced by topological recursion of the spectral curve
\be
\omega_{0,1}(z)=\frac{z-t}{z}\, \d z\ , \qquad \omega_{0,2}(z_1,z_2)=\frac{\d z_1\, \d z_2}{(z_1-z_2)^2}
\ee
and the single branch point located at $z=t$. It is therefore natural to expect that the super VMS amplitudes with a gas of Ramond defects are also dual to a deformed spectral curve. The parameter $t$ shifts the hard edge of the tree-level resolvent. We will not attempt to construct this modified spectral curve.
\section{Discussion} \label{sec:discussion}
We conclude by discussing a few future directions and connections to other works. 
\emph{Worldsheet analysis.} We have studied the super Virasoro minimal string from the point of view of 3d gravity and have said little about actual worldsheet explorations. Building on \cite{Muhlmann:2025ngz, Rangamani:2025wfa}, this perspective will be further developed in a forthcoming paper \cite{Muhlmann_toappear}.

\emph{Super JT gravity.} In the $b \to 0$ limit, super VMS reduces to super JT gravity studied in \cite{Stanford:2019vob}. This is manifest from the 3d gravity perspective: $b \to 0$ corresponds to $\hbar \to 0$ in the quantization procedure and the number of quantum states becomes the volume of phase space, i.e.\ supersymmetric generalizations of the Weil-Petersson volume. For the $-$ theories, one obtains the super volume form and super JT gravity on moduli space, while for the $+$ theories, the fermionic directions only supply a factor of $2^{\rk \mathcal{F}}$ and the semiclassical limit should be understood in terms of bosonic JT gravity. Many of the statements made in this paper have an analogous statement in super JT gravity, although with reversed 0A/0B conventions, see footnote \ref{footnote:0A/0B convention}.

\emph{Super topological recursion.} Recently a framework of topological recursion on super surfaces was developed \cite{Bouchard:2020pji, Aghaei:2025lhy}. Since we reduced all expressions to ordinary moduli space, we do not make use of this formalism, but it could perhaps constitute an alternative route for supersymmetrizing the VMS dualities. 

\emph{Perturbative vanishing of 0A$^-$.} As mentioned above, all perturbative amplitudes (apart from the sphere two-point function) vanish identically (and non-trivially) in type 0A$^-$. This exact perturbative vanishing is a rare phenomenon in string theory. Another example with a related property is the $\cN=2$ string/self-dual gravity \cite{Ooguri:1990ww, Ooguri:1995cp}. It would be interesting to understand whether there is a general underlying mechanism guaranteeing such vanishing theorems.

\emph{Ramond gas for 0B$^-$.} Type 0B$^-$ features also Ramond punctures. In the worldsheet description, their amplitudes vanish identically. From intersection theory, it is however possible to define non-vanishing amplitudes with Ramond-punctures and we argued that the discrepancy can be traced back to an operator normalization in the worldsheet theory. The intersection expressions for amplitudes with Ramond insertions are not naturally described by the matrix integral. Instead it is more natural to consider a `gas' of such punctures, i.e.\ turn on an R-R background for the model. It is natural to conjecture that such a gas of Ramond punctures deforms the spectral curve, as is also conjectured for JT supergravity in \cite{Stanford:2019vob, Alexandrov:2024kuj}.

\emph{Super complex Liouville string.} One can similarly wonder about supersymmetric generalizations of the complex Liouville string. There are also four theories 0A$^+$, 0B$^+$, 0A$^-$ and 0B$^-$ in that case. Based on a worldsheet bootstrap, \cite{Du:2025lya} made a proposal for a (two-) matrix integral dual which bears striking resemblance to the super VMS case discussed in this paper. It would be interesting to understand those dualities also from a more conceptual viewpoint.

\emph{$\cN=2$ super VMS.} One can try to further generalize to $\cN=2$ worldsheet supersymmetry. This is technically difficult from the worldsheet since spacelike and timelike $\cN=2$ theories are much less understood than their $\cN=0$ and $\cN=1$ counterparts. Nonetheless, it may be possible to still establish a duality by studying the moduli space of $\cN=2$ super Riemann surfaces. The classical limit is expected to reduce to $\cN=2$ JT supergravity, which was studied in \cite{Turiaci:2023jfa}, where the dual spectral curve was identified.

\emph{Heterotic VMS.} Another generalization is a variant inspired by \cite{Davis:2005qe, Seiberg:2005nk}. Its worldsheet theory is given by coupling timelike $\cN=1$ Liouville theory to spacelike $\cN=1$ Liouville theory with 22 left-moving fermions, right-moving $\beta\gamma$ ghosts and the usual $\mathfrak{b}\mathfrak{c}$ ghosts. Since this worldsheet theory is not left-right symmetric, the relation to 3d gravity does not simply generalize. Thus it constitutes an interesting generalization pushing the boundaries of the minimal string/matrix integral paradigm.

\section*{Acknowledgements} We thank Rafael \'Alvarez-Garc\'ia, Scott Collier, Duarte Fragoso, Alessandro Giacchetto, Mukund Rangamani, Bob Knighton, Beatrix M\"uhlmann, Victor Rodriguez and Ioannis Tsiares for useful conversations. We thank Scott Collier, Bob Knighton, Mukund Rangamani, Victor Rodriguez for comments on an early draft. LE is supported by the European Research Council (ERC) under the European Union’s Horizon 2020 research and innovation programme (grant agreement No 101115511). This paper made use of `Get Physics Done' for adversarial review \cite{physical_superintelligence_2026_gpd}. All computations and writing were carried out by LE.

\appendix 
\section{Topological recursion with hard edges} \label{app:hard topological recursion}
The topological recursion encountered for the 0B$^-$ case involves the Theta class. Both appendices are mathematically rigorous, which is why we switch to a mathematical writing style. Let us state the main proposition. 
\begin{proposition} \label{prop:hard topological recursion}
Consider a spectral curve of the form $\Sigma=\CP^1$ with $\omega_{0,1}(z)=\sum_{k=0}^\infty y_{k} z^{2k} \d z $ with $y_0=1$ and $\omega_{0,2}(z_1,z_2)=\frac{\d z_1\, \d z_2}{(z_1-z_2)^2}$. Define $\omega_{g,n}$ as usual, see \eqref{eq:topological recursion}. Then
\be
\omega_{g,n}(\boldsymbol{z})=\int_{\bM_{g,n}} \Tgn \mathrm{e}^{\sum_j s_j \kappa_j}\prod_{i=1}^n \sum_{d\ge 0} \frac{(2d+1)!! \psi_i^d\, \d z_i}{z_i^{2d+2}}\ .
\ee
Here, $s_j$ are defined through the equality of the following power series
\be 
\exp\bigg(-\sum_{j=1}^\infty s_j u^j \bigg)=\sum_{k=0}^\infty y_{k} (2k-1)!! u^k \label{eq:power series identification}
\ee
with $(-1)!!=1$.
\end{proposition}
This is very similar to the analogous theorem for the case $y_0=0$, see \cite{Eynard:2011ga}.

Let us first give some intuition about this. Imposing that $\omega_{0,1}(z)$ starts with a constant term implies a hard edge, i.e.\ inverse square root singularity in the density of states of the matrix integral, as opposed to the more conventional square root singularity. This modification makes the recursion kernel \eqref{eq:recursion kernel} nonsingular. As a consequence all $\omega_{0,n}$ with $n \ge 3$ vanish identically. However, once one computes $\omega_{1,1}$, the handle pinching term in the topological recursion \eqref{eq:topological recursion} is non-zero and higher genus data is non-vanishing.

\begin{proof}
    As far as we are aware, Proposition~\ref{prop:hard topological recursion} has not been stated in this form in the literature, but it is a simple consequence of two known facts. \cite[Proposition 4.2]{Chekhov:2017aot} establishes
    \begin{align}
        \omega_{g,n}(z_1,\dots,z_n)=\sum_{m=0}^{g-1} \sum_{\boldsymbol{d} \ge 0,\,\boldsymbol{\alpha}\ge 1} \prod_{i=1}^n \frac{\d z_i}{z_i^{2d_i+2}} \frac{(-1)^m}{m!}\, b_{g,n+m}(\boldsymbol{d},\boldsymbol{\alpha}) \prod_{k=1}^m \frac{y_{\alpha_{k}}}{2\alpha_k+1}\ , \label{eq:omega Bessel curve coefficients}
    \end{align}
where $b_{g,n}$ appears as the coefficients of the Bessel curve for which $\omega_{0,1}(z)=\d z$,
\be 
\omega_{g,n}^\text{Bes}(\boldsymbol{z})=\sum_{\boldsymbol{d} \ge 0} b_{g,n}(\boldsymbol{d}) \prod_{i=1}^n \frac{\d z_i}{z_i^{2d_i+2}}\ .
\ee
The Bessel spectral curve is the analogue of the Airy spectral curve and is dual to the intersection numbers of $\Tgn$ together with gravitational descendants, i.e.\ $\psi$-class insertions. Concretely, see \cite[Theorem 4.11]{Chidambaram:2022cqc}
\be 
\omega_{g,n}^\text{Bes}(\boldsymbol{z})=\sum_{\boldsymbol{d} \ge 0} \int_{\bM_{g,n}} \Tgn \prod_{i=1}^n \frac{(2d_i+1)!! \psi_i^{d_i} \d z_i}{z_i^{2d_i+2}}\ . \label{eq:Bessel curve correlators}
\ee
This is the special case of the proposition for $\omega_{0,1}(z)=\d z$.

Combining \eqref{eq:omega Bessel curve coefficients} and \eqref{eq:Bessel curve correlators} gives
\be 
\omega_{g,n}(\boldsymbol{z})=\sum_{m=0}^{g-1} \sum_{\boldsymbol{\alpha}\ge 1} \frac{(-1)^m}{m!}\int_{\bM_{g,n+m}} \Theta_{g,n+m}\prod_{i=1}^n \d\eta(z_i,\psi_i)\prod_{k=1}^m \big(y_{\alpha_k}(2\alpha_k-1)!!\psi_{n+k}^{\alpha_k}\big)\ , \label{eq:omega as psi classes}
\ee
where $\d \eta(z,\psi)=\sum_{d \ge 0} (2d+1)!! \psi^d z^{-2d-2} \d z$ is the differential that appears in the proposition. 

It remains to turn $\psi$-class intersections into $\kappa$-class intersections. One can do this systematically by using \cite[Theorem 5.7]{Norbury:2020vyi}. It gives the identity\footnote{The theorem is stated in terms of the Br\'ezin-Gross-Witten tau function $Z^\Theta$. To obtain the statement displayed here, one can restrict to a fixed genus and extract the coefficient of $\prod_{i=1}^n t_{d_i}$ on both sides.}
\be
\int_{\bM_{g,n}} \Tgn\, \mathrm{e}^{\sum_{j \ge 1} s_j \kappa_j} \prod_{i=1}^n \psi_i^{d_i}=\sum_{m \ge 0} \frac{1}{m!} \sum_{\boldsymbol{\alpha} \ge 1} \int_{\bM_{g,n+m}} \Theta_{g,n+m} \prod_{i=1}^n \psi_i^{d_i} \prod_{k=1}^m \sigma_{\alpha_k}\, \psi_{n+k}^{\alpha_k}\ , \label{eq:Norbury n external}
\ee
where $\sigma_k$ is defined through equality of the formal power series
\be
1-\exp\bigg(-\sum_{j=1}^\infty s_j u^j \bigg)=\sum_{k=1}^\infty \sigma_k u^k\ .
\ee
Thus, if we identify $\sigma_k=-(2k-1)!! y_k$, the right hand side of \eqref{eq:Norbury n external} has the form of the right hand side of \eqref{eq:omega as psi classes}. This proves the proposition.
\end{proof}

\section{Vanishing of the Arf Theta class} \label{app:arf theta class}
In this appendix, we analyze the Arf theta class $\Tgn^\zeta$ and show that all intersection numbers of all gravitational descendants vanish.\footnote{I am grateful to Alessandro Giacchetto for correspondence. He independently observed and proved proposition \ref{prop:Arf Theta vanishing} in \cite{Giacchetto:unpublished}.} 
\begin{proposition} \label{prop:Arf Theta vanishing}
    Define the Arf Theta class as
    \be 
        \Tgn^\zeta=2^{g-1}(-2)^n \big(p_*^+ c(\cF)-p_*^- c(\cF)\big)|_{2g-2+n}\ ,
    \ee
    where $\cF$ is the fermionic tangent space of supermoduli space and $p_*^\pm$ is the forgetful pushforward from the spin moduli space with even (odd) spin structure. Then any intersection number of the following form vanishes,
    \be 
    \int_{\bM_{g,n}}\Tgn^\zeta\eta=0\ ,
    \ee 
    where $\eta$ is an arbitrary product of $\psi$ and $\kappa$ classes.
\end{proposition}

We will prove Proposition \ref{prop:Arf Theta vanishing} by employing the following strategy. We first establish that $\Tgn^\zeta$ behaves naturally under pullback. Mathematically, it defines a cohomological field theory with a 1-dimensional state space. The flat unit axiom is replaced by the pullback in Lemma~\ref{lem:unit}. As shown by Norbury \cite{Norbury:2017eih}, these properties characterize intersection numbers of $\Tgn^\zeta$ uniquely up to normalization. An explicit computation shows that the normalization is zero, which forces vanishing of all intersection numbers.
\begin{lemma}
    The weighted pushforward of the total Chern class admits the following explicit expression in $H^\bullet(\bM_{g,n})$ 
    \begin{multline} 
        \Tgn^\zeta=(-2)^{2g-2+n}\exp\bigg(-\sum_{d \ge 1} \frac{1}{d(d+1)}\bigg[B_{d+1}(\tfrac{3}{2})\kappa_d-\sum_{i=1}^n B_{d+1}(\tfrac{1}{2}) \psi_i^d\\
        +B_{d+1}(\tfrac{1}{2}) \xi_* \Big(\frac{\psi_\bullet^d-(-\psi_\circ)^d}{\psi_\bullet+\psi_\circ}\Big)\bigg]\bigg)\bigg|_{2g-2+n}\ . \label{eq:Thetazeta exponential}
    \end{multline}
    Here, $\xi$ is the inclusion of all boundary divisors into moduli space, i.e.\ 
    \be 
    \xi_*=\frac{1}{2}\sum_{h=0}^g \sideset{}{'}\sum_{I \subset [n]} (\xi_{h,I})_*+\frac{1}{2} (\xi_\mathrm{irr})_*\ .
    \ee
    The factor of $\frac{1}{2}$ avoids overcounting. The prime on the sum restricts the sum to stable terms.
\end{lemma}
\begin{proof}
From the Chiodo formula \eqref{eq:Chiodo formula}, we have
\be 
\ch_d(\cF)=\frac{(-1)^d}{(d+1)!} \bigg(B_{d+1}(\tfrac{3}{2})\kappa_d-\sum_{i=1}^n B_{d+1}(\tfrac{1}{2}) \psi_i^d+\sum_{q=0,1} B_{d+1}(\tfrac{q}{2}) (\xi_q)_* \Big(\frac{\psi_\bullet^d-(-\psi_\circ)^d}{\psi_\bullet+\psi_\circ}\Big)\bigg)
\ee
in $H^\bullet(\bM_{g,n}^\text{spin})$. We can translate this to the total Chern class and obtain the formula
\begin{align}
    c(\cF)&=\exp\bigg(-\sum_{d\ge 1} (-1)^d (d-1)!\, \ch_d(\cF)\bigg) \\
    &=\exp\bigg(-\sum_{d \ge 1} \frac{1}{d(d+1)}\bigg[B_{d+1}(\tfrac{3}{2})\kappa_d-\sum_{i=1}^n B_{d+1}(\tfrac{1}{2}) \psi_i^d\nonumber\\
    &\qquad\qquad\qquad\qquad+\sum_{q=0,1} B_{d+1}(\tfrac{q}{2}) (\xi_q)_* \Big(\frac{\psi_\bullet^d-(-\psi_\circ)^d}{\psi_\bullet+\psi_\circ}\Big)\bigg]\bigg)\ .
\end{align}
We now need to push this forward to moduli space while accounting for the weighting.
The class $c(\cF)$ admits an expansion as a sum over weighted stable graphs $\Gamma$, see \cite[Proposition 4]{JPPZ}. This graph expansion is obtained by repeatedly pulling back to the boundary strata and using the excess formula. In the process one uses that the Chern class of the normal bundle to a boundary divisor is $c_1(\cN)=-\frac{1}{2}(\psi_\bullet+\psi_\circ)$. The factor of $\frac{1}{2}$ is due to a $\ZZ_2$ automorphism factor of the boundary stratum in spin moduli space. We obtain
\begin{align}
c(\cF)&=\sum_{\Gamma \in \mathcal{G}_{g,n}} \sum_{w \in \mathcal{W}_{\Gamma}}\frac{2^{|\mathcal{E}_\Gamma|}}{|\Aut \Gamma|} (\xi_{\Gamma,w})_*\bigg[ \prod_{v \in \mathcal{V}_\Gamma} \exp\bigg(-\sum_{d \ge 1}\frac{B_{d+1}(\frac{3}{2}) \kappa_d}{d(d+1)} \bigg) \nonumber\\
&\qquad\times \prod_{i=1}^n \exp\bigg(\sum_{d \ge 1} \frac{B_{d+1}(\frac{1}{2}) \psi_i^d}{d(d+1)}\bigg)\nonumber\\
&\qquad\prod_{e=(\bullet,\circ) \in \mathcal{E}_\Gamma} \frac{1-\exp\Big(\sum_{d \ge 1} \frac{B_{d+1}(\frac{w(e)}{2})}{d(d+1)} \big(\psi_\bullet^d-(-\psi_\circ)^d\big)\Big)}{\psi_\bullet+\psi_\circ} \bigg]\ . \label{eq:c(F) graph expansion}
\end{align}
Here, the weighting corresponds to an assignment $w(e) \in \{0,1\}$ of NS- or R-sector to each edge. In every vertex, the sum of outgoing Ramond edges has to be even. The map $\xi_{\Gamma,w}$ is the inclusion of the stratum corresponding to the weighted graph $(\Gamma,w)$ into $\bM_{g,\bsa}^\text{spin}$.

We now need to push \eqref{eq:c(F) graph expansion} forward to moduli space and weigh with the Arf invariant. This can be done by following the proof of \cite[Proposition 9.21]{Giacchetto:2021qss}. In the process, graphs with Ramond edges will cancel, essentially since for them the Arf invariant can be continuously changed. Thus we can restrict to weightings where $w(e)=1$ for all edges, i.e.\ we have only NS edges. This is intuitive since it says that the string theory should not have propagating Ramond-states which are not part of the spectrum. In each graph, we can then push forward to moduli space with the inclusion of the Arf invariant. This gives a factor of $\prod_{v \in \mathcal{V}} 2^{g_v-1}$, where $g_v$ is the genus associated to a vertex, since the difference of even and odd spin structures in each vertex is $2^{g_v}$ and the extra $\frac{1}{2}$ accounts for the $\ZZ_2$ automorphism factor of each vertex. We have
\be 
\sum_{v \in \mathcal{V}_\Gamma} (g_v-1)=\sum_v g_v-|\mathcal{V}_\Gamma|=g-h_1(\Gamma)-(|\mathcal{E}_\Gamma|+1-h_1(\Gamma))=g-1-|\mathcal{E}_\Gamma|\ .
\ee
Here, $h_1(\Gamma)=g-\sum_v g_v$ is the number of loops in the stable graph. We also used Euler's formula $h_1(\Gamma)=|\mathcal{E}_\Gamma|-|\mathcal{V}_\Gamma|+1$. The factor $2^{-|\mathcal{E}_\Gamma|}$ compensates exactly a factor $2^{|\mathcal{E}_\Gamma|}$ present in \eqref{eq:c(F) graph expansion}. Thus
\begin{align} 
p_*^+ c(\cF)-p_*^- c(\cF)&=\sum_{\Gamma \in \mathcal{G}_{g,n}} \frac{2^{g-1}}{|\Aut \Gamma|} (\xi_\Gamma)_*\bigg[ \prod_{v \in \mathcal{V}_\Gamma} \exp\bigg(-\sum_{d \ge 1}\frac{B_{d+1}(\frac{3}{2}) \kappa_d}{d(d+1)} \bigg) \nonumber\\
&\qquad\prod_{i=1}^n \exp\bigg(\sum_{d \ge 1} \frac{B_{d+1}(\frac{1}{2}) \psi_i^d}{d(d+1)}\bigg) \nonumber\\
&\qquad\times \prod_{e=(\bullet,\circ) \in \mathcal{E}_\Gamma} \frac{1-\exp\Big(\sum_{d \ge 1} \frac{B_{d+1}(\frac{1}{2})}{d(d+1)} \big(\psi_\bullet^d-(-\psi_\circ)^d\big)\Big)}{\psi_\bullet+\psi_\circ} \bigg]\ .
\end{align}
As a final step, one can undo the graph expansion by the same logic that led to \eqref{eq:c(F) graph expansion}. This is now on ordinary moduli space and the first Chern class of the normal bundle to the boundary divisor is $c_1(\cN)=-\psi_\bullet-\psi_\circ$. Therefore, this does not lead to further factors of two and proves the lemma.
\end{proof}
\begin{lemma}\label{lem:CohFT}
    $\Tgn^\zeta$ satisfies the factorization axioms of a one-dimensional cohomological field theory, i.e.\ 
    \be 
    \xi_{h,I}^*\Tgn^\zeta=\Theta_{h,|I|+1}^\zeta \boxtimes \Theta_{g-h,|I^c|+1}^\zeta\ , \qquad \xi_\mathrm{irr}^*\Tgn^\zeta=\Theta_{g-1,n+2}^\zeta\ . \label{eq:CohFT Thetazeta}
    \ee
\end{lemma}
\begin{proof}
    The proof of the two equations is essentially identical. Let us explain the proof of the first one. Write
    \be 
        \alpha_{g,n}=-\sum_{d \ge 1} \frac{1}{d(d+1)}\bigg[B_{d+1}(\tfrac{3}{2})\kappa_d-\sum_{i=1}^n B_{d+1}(\tfrac{1}{2}) \psi_i^d\\
        +B_{d+1}(\tfrac{1}{2}) \xi_* \Big(\frac{\psi_\bullet^d-(-\psi_\circ)^d}{\psi_\bullet+\psi_\circ}\Big)\bigg] \ . \label{eq:alpha gn definition}
    \ee
    Since the exponential commutes with the pullback, it will be sufficient to compute $\xi_{h,I}^* \alpha_{g,n}$. We recall the basic pullbacks
    \begin{subequations}
    \begin{align} 
        \xi_{h,I}^*\kappa_d&=\kappa_d \boxtimes 1+1 \boxtimes \kappa_d\ , \label{eq:boundary pullback kappa classes}\\
        \xi_{h,I}^* \psi_i&=\begin{cases}
            \psi_i \boxtimes 1\ , &i \in I\ , \\
            1 \boxtimes \psi_i\ , &i \in I^c\ .
        \end{cases} \label{eq:boundary pullback psi classes}
    \end{align}
    \end{subequations}
    where the superscripts refer to the two components of moduli space. We can pull back $(\xi_{h',I'})_*\beta$ and $(\xi_\mathrm{irr})_*\beta$ for some class $\beta$ as described e.g.\ in \cite[Corollary 3.7]{Bae:2022prestable}. For $(h',I') \ne (h,I)$ as well as for the non-separating divisor, the result is the pushforward over all graphs that are in the intersection of the divisors $\mathcal{D}_{h,I}$ and $\mathcal{D}_{h',I'}$ or $\mathcal{D}_\text{irr}$, respectively. 
    To state the precise formula, let us write 
    \be 
        (h',I') \le  (h,I)  \quad\Longleftrightarrow\quad h' \le h \text{ and }I' \subseteq I\text{ or }g-h' \le h \text{ and }(I')^c \subseteq I\ .
    \ee
    We also write $(h',I')<(h,I)$ when $(h',I') \le (h,I)$ and $(h',I') \ne (h,I)$. Then for $k \ge 0$ and $\ell \ge 0$
    \begin{subequations}
\begin{align}
\xi_{h,I}^*\big((\xi_{h',I'})_*(
  \psi_\bullet^k\psi_\circ^\ell)\big)
&=
\begin{cases}
-\psi_{|I|+1}^{k+1}\boxtimes \psi_{|I^c|+1}^\ell
-\psi_{|I|+1}^k \boxtimes \psi_{|I^c|+1}^{\ell+1}
& \!\text{if } (h',I')=(h,I)\ ,
\\
(\xi_{h',I'})_*(
  \psi_\bullet^k\psi_\circ^\ell) \boxtimes 1 & \!\text{if } (h',I')<(h,I)\ , \\
  1 \boxtimes (\xi_{h',I'})_*(
  \psi_\bullet^k\psi_\circ^\ell) & \!\text{if } (h',I')<(g-h,I^c)\ , \\
  0 & \!\text{otherwise}\ ,
\end{cases} \label{eq:boundary pullback separating boundary pushforward}\\
\xi_{h,I}^*\big((\xi_\mathrm{irr})_*(\psi_\bullet^k \psi_\circ^\ell)\big)&=(\xi_\mathrm{irr})_*(\psi_\bullet^k \psi_\circ^\ell) \boxtimes 1 +1 \boxtimes (\xi_\mathrm{irr})_*(\psi_\bullet^k \psi_\circ^\ell)\ .
\end{align}\label{eq:boundary pullback boundary pushforward}%
\end{subequations} 
The first case in \eqref{eq:boundary pullback separating boundary pushforward} is the excess formula and captures the self-intersection of the boundary divisor $\mathcal{D}_{h,I}$.
In writing the second case of \eqref{eq:boundary pullback separating boundary pushforward}, we can replace $(h',I')$ by $(g-h',I^c)$ if necessary and assume that $h' \le h$ and $I' \subseteq I$. Similarly in the third case, we can assume that $h' \le g-h$ and $I' \subseteq I^c$. 

With \eqref{eq:boundary pullback kappa classes}, \eqref{eq:boundary pullback psi classes} and \eqref{eq:boundary pullback boundary pushforward}, it is simple to compute the pullback of $\alpha_{g,n}$. Since $(h',I')$ in \eqref{eq:boundary pullback separating boundary pushforward} runs over all divisors of $\bM_{h,|I|+1}$ and $\bM_{g,|I^c|+1}$, this produces exactly the expected boundary contributions. We find
    \begin{align}
        \xi_{h,I}^*\alpha_{g,n}&=-\sum_{d \ge 1} \frac{1}{d(d+1)} \bigg[ B_{d+1}(\tfrac{3}{2})(\kappa_d \boxtimes 1+1 \boxtimes \kappa_d)-\sum_{i \in I} B_{d+1}(\tfrac{1}{2})\psi_i^d \boxtimes 1\nonumber\\
        &\qquad-\sum_{i \in I^c} B_{d+1}(\tfrac{1}{2}) 1 \boxtimes \psi_i^d-B_{d+1}(\tfrac{1}{2}) (\psi_{|I|+1}^d \boxtimes 1+1 \boxtimes \psi_{|I^c|+1}^d)\nonumber\\
        &\qquad+B_{d+1}(\tfrac{1}{2}) \bigg(\xi_* \Big(\frac{\psi_\bullet^d-(-\psi_\circ)^d}{\psi_\bullet+\psi_\circ}\Big) \boxtimes 1+1 \boxtimes \xi_* \Big(\frac{\psi_\bullet^d-(-\psi_\circ)^d}{\psi_\bullet+\psi_\circ}\Big)\bigg)\bigg]\\
        &=\alpha_{h,|I|+1} \boxtimes 1 +1 \boxtimes \alpha_{g-h,|I^c|+1}\ .
    \end{align}
    The $\psi$-classes for the nodes that appear in the second lines arise from the self-intersection case of the boundary classes \eqref{eq:boundary pullback separating boundary pushforward}.
    
    It remains to take the exponential and restrict to the pure degree $2g-2+n$. On the right hand side of the first equation of \eqref{eq:CohFT Thetazeta}, the degree can be in principle be distributed in any way on the two factors. However, $\Tgn^\zeta$ is defined as the pushforward of the total Chern class of a rank $2g-2+n$ bundle. Therefore degrees higher than $2g-2+n$ vanish by construction. The only non-vanishing degree assignment on the RHS is then 
    \be 
        \xi_{h,I}^* \exp(\alpha_{g,n}) \big|_{2g-2+n}= \exp(\alpha_{h,|I|+1}) \big|_{2h-1+|I|} \boxtimes \exp(\alpha_{g-h,|I^c|+1}) \big|_{2(g-h)-1+|I^c|}\ ,
    \ee
    which is the claimed equation. The proof for the pullback to the non-separating divisor is essentially identical and we omit the details.
\end{proof}

\begin{lemma} \label{lem:unit}
    $\Tgn^\zeta$ satisfies
    \be 
        \psi_{n+1}\pi^*\Tgn^\zeta=\Theta_{g,n+1}^\zeta\ .
    \ee
    Here, $\pi: \bM_{g,n+1} \longrightarrow \bM_{g,n}$ is the forgetful morphism. 
\end{lemma}
\begin{proof}
    Let us compute $\pi^*\alpha_{g,n}$ for $\alpha_{g,n}$ defined in \eqref{eq:alpha gn definition}. We use the elementary pullbacks
    \begin{subequations}
    \begin{align}
        \pi^* \kappa_d&=\kappa_d-\psi_{n+1}^d\ , \label{eq:kappad forgetful pullback}\\
        \pi^* \psi_i^d&=\psi_i^d-(\xi_{0,\{i,n+1\}})_*(\psi_\circ^{d-1})\ . \label{eq:psi^d forgetful pullback}
    \end{align}
    \end{subequations}
    Here, $\psi_\circ$ in the second line is the $\psi$-class associated to the node of the second factor of $\bM_{0,3} \times \bM_{g,n}\subset \bM_{g,n+1}$. Since $\bM_{0,3}$ has vanishing dimension, the $\psi$-class on the $\bM_{0,3}$ factor does not appear.
We also have
\begin{align}
    \pi^* ((\xi_{h,I})_*\psi_\bullet^k \psi_\circ^\ell)&=(\xi_{h,I \cup \{n+1\}})_*\big(\pi^*(\psi_\bullet^k) \psi_\circ^\ell \big)+(\xi_{h,I})_*\big(\psi_\bullet^k \pi^*(\psi_\circ^\ell) \big) \\
    &=(\xi_{h,I \cup \{n+1\}})_* (\psi_\bullet^k \psi_\circ^\ell)+(\xi_{h,I})_* (\psi_\bullet^k \psi_\circ^\ell)-(\xi_{\Gamma_{h,I}})_*\big(\psi_\bullet^{k-1} \psi_\circ^\ell+\psi_\bullet^{k} \psi_\circ^{\ell-1}\big)\ .
\end{align}
Here
\be 
\Gamma_{h,I}=\begin{tikzpicture}[baseline={([yshift=-.5ex]current bounding box.center)},scale=.6]
    \node[shape=circle,draw=black, very thick, fill=white, minimum size=6mm, inner sep=0pt] (A) at (-2.5,0) {$h$};
    \node[shape=circle,draw=black, very thick, fill=white, minimum size=6mm, inner sep=0pt] (B) at (0,0) {$0$};
    \node[shape=circle,draw=black, very thick, fill=white, minimum size=6mm, inner sep=0pt] (C) at (2.5,0) {$h'$};
    \draw[thick] (A) to node[above] {$\psi_\bullet\quad$} (B);
    \draw[thick] (B) to node[above] {$\quad\psi_\circ$} (C);
    \draw[thick] (-4,.8) to (A);
    \draw[thick] (-4,.6) to (A);
    \node at (-3.5,.1) {$\vdots$};
    \draw[thick] (-4,-.8) to (A);
    \draw[thick] (4,.8) to (C);
    \draw[thick] (4,.6) to (C);
    \node at (3.5,.1) {$\vdots$};
    \draw[thick] (4,-.8) to (C);
    \node at (-4.3,0) {$I$};
    \node at (4.3,0) {$I^c$};
    \draw[thick] (B) to (0,1.3) node[above] {$n+1$};
\end{tikzpicture}\ ,
\ee
where $h'=g-h$.
We abuse notation slightly, since $\psi_\bullet$ and $\psi_\circ$ in the last term are not associated to the same edge, but instead of the indicated half-legs in the stable graph. The term involving $\psi_\bullet^{k-1}$ is absent when $k=0$ and similarly for $\psi_\circ^{\ell-1}$.

We now apply this pullback to the boundary pushforward of the class $\frac{\psi_\bullet^d-(-\psi_\circ)^d}{\psi_\bullet+\psi_\circ}=\sum_{k+\ell=d-1} \psi_\bullet^k(-\psi_\circ)^\ell$ that appears in \eqref{eq:alpha gn definition} with the result
\begin{multline} 
\pi^*\bigg((\xi_{h,I})_*\Big(\frac{\psi_\bullet^d-(-\psi_\circ)^d}{\psi_\bullet+\psi_\circ}\Big) \bigg)=\big((\xi_{h,I})_*+(\xi_{h,I \cup \{n+1\}})_*\big)\Big(\frac{\psi_\bullet^d-(-\psi_\circ)^d}{\psi_\bullet+\psi_\circ}\Big) \\
-\sum_{k+\ell=d-1} (-1)^\ell(\xi_{\Gamma_{h,I}})_*\big(\psi_\bullet^{k-1} \psi_\circ^\ell+\psi_\bullet^{k} \psi_\circ^{\ell-1}\big)\ . \label{eq:separating boundary pushforward forgetful pullback}
\end{multline}
Due to the alternating sign, the last term is a telescoping sum and cancels out. The boundary terms in the alternating sum are by definition absent. We similarly find
\be 
\pi^*\bigg((\xi_{\text{irr}})_*\Big(\frac{\psi_\bullet^d-(-\psi_\circ)^d}{\psi_\bullet+\psi_\circ}\Big) \bigg)=(\xi_{\mathrm{irr}})_*\Big(\frac{\psi_\bullet^d-(-\psi_\circ)^d}{\psi_\bullet+\psi_\circ}\Big)\ ,\label{eq:non-separating boundary pushforward forgetful pullback}
\ee
with the codimension-2 boundary terms cancelling.

It now almost follows that the boundary contribution of \eqref{eq:alpha gn definition} pulls back to the corresponding contribution on $\bM_{g,n+1}$. Indeed, the pushforwards $(\xi_{h,I})_*$ and $(\xi_{h,I \cup \{n+1\}})_*$ for $I \subset [n]$ and $0 \le h \le g$ subject to stability appearing in \eqref{eq:separating boundary pushforward forgetful pullback} as well as $(\xi_\mathrm{irr})_*$ appearing in \eqref{eq:non-separating boundary pushforward forgetful pullback} account for almost all boundary pushforwards. The only exception is the pushforward $(\xi_{0,\{i,n+1\}})_*$ for $i \in \{1,\dots,n\}$. It is instead produced by the pullback of the class $\psi_i$. So in the pullback of $\alpha_{g,n}$, we find all the terms of $\alpha_{g,n+1}$, except for the $\psi$-class correction of the pullback of the $\kappa$-classes \eqref{eq:kappad forgetful pullback} and the class $\psi_{n+1}^d$ that is missing from the second term of $\pi^* \alpha_{g,n}$.
With this, we find
    \be 
        \pi^*\alpha_{g,n}-\alpha_{g,n+1}=\sum_{d \ge 1} \frac{B_{d+1}(\frac{3}{2})-B_{d+1}(\frac{1}{2})}{d(d+1)} \psi_{n+1}^d \\
        =\sum_{d \ge 1} \frac{\psi_{n+1}^d}{2^d d}=-\log\big(1-\tfrac{1}{2} \psi_{n+1}\big)\ .
    \ee
We used that $B_{d+1}(\frac{3}{2})-B_{d+1}(\frac{1}{2})=(d+1)2^{-d}$, which follows directly from the generating function of Bernoulli numbers.
    Thus
    \be 
        \big(1-\tfrac{1}{2}\psi_{n+1}\big) \pi^* \exp(\alpha_{g,n})=\exp(\alpha_{g,n+1})\ .
    \ee
    We can lastly restrict this to degree $2g-1+n$. On the LHS, only $-\frac{1}{2} \psi_{n+1} \pi^* \exp(\alpha_{g,n})$ survives, since $\exp(\alpha_{g,n})$ doesn't have support in degrees higher than $2g-2+n$. With the prefactor $(-2)^{2g-2+n}$ in \eqref{eq:Thetazeta exponential}, the lemma follows.
\end{proof}
\begin{proof}[Proof of Proposition \ref{prop:Arf Theta vanishing}]
    We can now finish the initial proposition. The point of the previous lemmas was to show that $\Tgn^\zeta$ satisfies the same properties as Norbury's Theta class \cite{Norbury:2017eih}. We can then use \cite[Proposition 3.2]{Norbury:2017eih}, whose hypotheses are: (i) $\Theta_{g,n}^\zeta$ is of pure degree $2g-2+n$, (ii) the CohFT factorization axioms (Lemma~\ref{lem:CohFT}), and (iii) the modified unit/dilaton axiom (Lemma~\ref{lem:unit}). All three have been verified above. This proposition shows that these properties determine any intersection number 
    \be 
        \int_{\bM_{g,n}} \Tgn^\zeta \prod_{i=1}^n \psi_i^{d_i} \prod_{d \ge 1} \kappa_d^{m_d} \label{eq:Thetagn zeta general intersection number}
    \ee
    recursively in terms of the initial value $\int_{\bM_{1,1}}\Theta_{1,1}^\zeta$. A direct computation gives
    \be 
        \int_{\bM_{1,1}} \Theta_{1,1}^\zeta =\frac{-2}{24}\int_{\bM_{1,1}}\Big(-11 \kappa_1-\psi_1+\frac{1}{2}(\xi_\mathrm{irr})_*(1)\Big)=0\ , 
    \ee
    where we used $\int_{\bM_{1,1}} \kappa_1=\int_{\bM_{1,1}} \psi_1=\frac{1}{24}$ and $\int_{\bM_{1,1}}(\xi_\mathrm{irr})_*(1)=\int_{\bM_{0,3}} 1=1$. Thus $\int_{\bM_{1,1}} \Theta_{1,1}^\zeta$ vanishes, which implies that all intersection numbers of the form \eqref{eq:Thetagn zeta general intersection number} also vanish.
\end{proof}
\bibliographystyle{JHEP}
\bibliography{bib}
\end{document}